\documentclass[pdflatex,sn-mathphys,Numbered,UKenglish]{sn-jnl}

\usepackage{graphicx}%
\usepackage{multirow}%
\usepackage{amsmath,amssymb,amsfonts}%
 \usepackage{amsthm}%
\usepackage{mathrsfs}%
\usepackage[title]{appendix}%
\usepackage{xcolor}%
\usepackage{textcomp}%
\usepackage{manyfoot}%
\usepackage{booktabs}%
\usepackage{algorithm}%
\usepackage{algorithmicx}%
\usepackage{algpseudocode}%
\usepackage{listings}%
\usepackage{lmodern}
\usepackage{anyfontsize}

\newtheorem{theorem}{Theorem}
\newtheorem{lemma}[theorem]{Lemma}%
\newtheorem{corollary}[theorem]{Corollary}%
\newtheorem*{theorem*}{Theorem}
\newtheorem{claim}{Claim}

\theoremstyle{definition}%
\newtheorem{definition}{Definition}%
\newtheorem{obs}{Observation}
\newtheorem{const}{Construction}
\usepackage{graphicx}
\usepackage{hyperref}
 \usepackage{cleveref}
\hypersetup{
    colorlinks=true,
    linkcolor=blue,
}

\usepackage{framed}

\usepackage[disableredefinitions]{complexity}
\usepackage{xspace}
\usepackage{booktabs}
\usepackage[T1]{fontenc}

\usepackage{mathtools}  

\usepackage{enumitem}

\usepackage{tikz}
\usetikzlibrary{positioning,quotes}

 \usepackage{subcaption}
\usepackage{algorithm}
\usepackage{algpseudocodex}

\DeclareMathOperator{\reach}{reach}
\usepackage{orcidlink}

\usepackage[disableredefinitions]{complexity}
\usepackage{xspace}
\usepackage{booktabs}

\usepackage{mathtools}  

\usepackage{tikz}
\usetikzlibrary{positioning,quotes}

\DeclareMathOperator{\tfo}{t_{\text{Fo}}}

\DeclareMathOperator{\teccs}{ecc^S}
\DeclareMathOperator{\teccfa}{ecc^{Fa}}

\newcommand{\trlpp}{TRLP\xspace}

\newcommand{\trppp}{TRP\xspace}

\usepackage{bbding}

\begin{document}

\title{Reachability in temporal graphs under perturbation} 

\author[1]{\fnm{Jessica} \sur{Enright}\,\orcidlink{0000-0002-0266-3292}}\email{jessica.enright@glasgow.ac.uk}
\author*[1]{\fnm{Laura} \sur{Larios-Jones}\,\orcidlink{0000-0003-3322-0176}}\email{l.larios-jones.1@research.gla.ac.uk}
\author[1]{\fnm{Kitty} \sur{Meeks} \,\orcidlink{0000-0001-5299-3073}}\email{kitty.meeks@glasgow.ac.uk}
\author[1]{\fnm{William} \sur{Pettersson}\,\orcidlink{0000-0003-0040-2088}}\email{william.pettersson@glasgow.ac.uk}

\affil[1]{\orgdiv{School of Computing Science}, \orgname{University of Glasgow}, \orgaddress{\country{Scotland}}}

%


\maketitle

\begin{abstract}

    Reachability and other path-based measures on temporal graphs can be used to understand spread of infection, information, and people in modelled systems. Due to delays and errors in reporting, temporal graphs derived from data are unlikely to perfectly reflect reality, especially with respect to the precise times at which edges appear. To reflect this uncertainty, we consider a model in which some number $\zeta$ of edge appearances may have their labels perturbed by $\pm\delta$ for some~$\delta$. Within this model, we investigate temporal reachability and consider the problem of determining the maximum number of vertices any vertex can reach under these perturbations. We show this problem to be intractable in general but efficiently solvable when $\zeta$ is sufficiently large. We also give algorithms which solve this problem in several restricted settings.  We complement this with some contrasting results concerning the complexity of related temporal eccentricity problems under perturbation.
\end{abstract}
\section{Introduction}
Temporal graphs are widely used to model movement and spread in time-sensitive networks \cite{casteigts_journey_2018,deligkas_being_2023,hand_making_2022,kutner_temporal_2023}.  However, the algorithmic tools used in solving problems on these graphs typically assume that the temporal graphs as given are correct and without uncertainty --- unfortunately this may not be the case for real-world networks (e.g. \cite{chaters2019analysing,myall2021characterising}), and thus this assumption significantly limits the applicability of temporal graph methods.  


For example, reachability is used as a measure to assess risk of outbreaks in epidemiologically-relevant graphs such as livestock trading networks or human or animal contact networks \cite{holme2005network,fielding2019contact}.  It is clear that a calculated reachability may be incorrect if the temporal graph itself is incorrect.  From an application perspective, the robustness of the temporal graph's reachability to incorrect timings is important: if a small number of incorrectly-recorded edge appearances can result in a huge increase in reachability, then this temporal graph is more epidemiologically risky than one in which a small number of incorrect appearance times do not significantly increase reachability.  

In this work we address the question: what is the complexity, given a specification of the number and magnitude of errors in the timings of the temporal graph's edge appearances, of determining if the reachability of the temporal graph could be above some given threshold?  Put another way, we ask if there exists a perturbation of the times allocated to edges that increases the reachability above a target threshold.  This second formulation makes clear a link to temporal graph modification questions --- an active area of research studying the minimum number of modifications required to a temporal graph that allow it to satisfy some specified property.

Following the definition given by Kempe, Kleinberg and Kumar \cite{kempe_connectivity_2002}, we say that temporal graphs $(G,\lambda)$ consist of an underlying graph $G=(V,E)$ and a temporal assignment $\lambda:E(G)\to 2^\mathbb{N}$ which describes when each edge is active. 
We introduce the concept of a $(\delta,\zeta)$-perturbation where up to $\zeta$ time-edges of a temporal graph are changed by at most $\pm\delta$. We ask, given a temporal graph $(G,\lambda)$, if there is a perturbation of $(G,\lambda)$ such that there is a temporal path from some vertex in the graph to at least $h$ vertices. With the framing of a transport network, this asks if we can reach $h$ places from some starting point by connections which occur chronologically.  In an epidemiological framework, this asks how many nodes could, in the worst case, be infected by a single initially-infected node.  


In previous work, Deligkas and Potapov~\cite{deligkas_optimizing_2022} consider delaying and merging operations to optimise maximum, minimum and average reachability where merging of time-edges assigns all edges which are active in a set of consecutive times the latest time in that set. They show that it is NP-hard to find a set of merging operations that maximises the maximum reachability of a temporal graph, even when the underlying graph is a path. 
In contrast, they show that finding a set of delaying operations to minimize the maximum reachability is tractable 
 if the number of delays made is unbounded. In a similar way, we find that allowing a large number of perturbations makes our problem easier.
 
In more recent work, Deligkas et al.~\cite{deligkas_minimizing_2023_bugfree} investigate the problem which asks how \emph{long} it takes to reach every vertex in a temporal graph from a set of sources following some delays. They consider variations of this problem where the number of delayed edges and total sum of delays must be at most $k$ for some integer $k$, and show that this problem is \W[2]{}-hard with respect to the number of edges changed and prove tractability when the underlying graph has bounded treewidth.
Delaying appearances of edges to optimise reachability is also studied by Molter et al.~\cite{molter_temporal_2021}, who show that the problem becomes fixed-parameter tractable with respect to the number of delayed edges.  Other work on temporal graph modification to optimise reachability has considered deleting edges or edge appearances \cite{molter_temporal_2021,enright_deleting_2021} and assigning times to edges in a static graph~\cite{enright_assigning_2021}.

Some related work exists on computing structures and properties which are robust to perturbation.  Fuchsle et al.~\cite{fuchsle_delay-robust_2022} ask if there is a path between two given vertices which is robust to delays.  In network design, fault-tolerance describes a network's robustness to uncertainty \cite{baswana_fault_2016,bilo_blackout-tolerant_2022,casteigts_robustness_2020,dinitz_fault-tolerant_2011,dubois_when_2023,lochet_fault_2020}. Other models containing uncertainty allow for queries to find true values and aim to optimise the number of queries needed to find a solution \cite{erlebach_learning-augmented_2022,erlebach_minimum_2014,erlebach_computing_2008,erlebach_round-competitive_2022}.

\subsection{Our Contributions}
Motivated by the idea of uncertainty in temporal graphs, we work with perturbations of temporal graphs, and focus on the computational complexity of a problem that asks whether there exists a perturbation of an input temporal graph such that there exists a source vertex that has at least a minimum specified reachability.  We:
\begin{itemize}
\item show that this problem is NP-hard in general, and \W[2]{}-hard with respect to the number of edges perturbed (Section~\ref{sec:intractability}),
\item show that the problem is tractable when the number of the perturbations allowed is at least the number of vertices we aim to reach minus one (Section~\ref{sec:many-perturbations}), and
\item show that the problem is tractable using a knapsack-based dynamic program when the underlying graph is a tree, or using a more complex tree decomposition-based algorithm when the underlying graph has bounded treewidth (Section~\ref{sec:structural}).  
\end{itemize}

To provide contrast with our results on the reachability problem, we briefly investigate (in Section~\ref{sec:eccentricity}) related eccentricity-based problems on temporal graphs with perturbations --- in these problems we require not only that vertices are reached from a source, but that they are reached within a given number of edges (shortest eccentricity) or within a given travel duration (fastest eccentricity).  In contrast to our findings for our reachability question, for fastest and shortest eccentricity we show that the problem remains hard when we are allowed to perturb all of the edges.

\subsection{Preliminaries}\label{sec:problems}
We give some notation and definitions, as well as including preliminary results relating to reachability in temporal graphs without perturbation.


A \emph{temporal graph} is a pair $(G,\lambda)$, where $G = (V,E)$ is an underlying (static) graph and $\lambda: E \rightarrow 2^{\mathbb{N}}$ is a time-labelling function which assigns to every edge of $G$ a set of discrete-time labels. When $|\lambda(e)|=1$, we may abuse notation and write $\lambda(e)=t$ rather than $\lambda(e)=\{t\}$ for some time $t$. An example temporal graph can be found in Figure~\ref{fig:tgraph-eg}. We call a pair $(e,t)$ where $t\in \lambda(e)$ a \emph{time-edge} and denote by $\mathcal{E}((G,\lambda))$ the set of all time-edges in $(G,\lambda)$. When the graph in question is clear from context, we denote the number of vertices in the graph $|V(G)|$ by $n$ and the number of edges in the underlying graph $|E(G)|$ by $m$.
We let $\tau(G,\lambda) \coloneqq \max_{e\in E(G)} |\lambda(e)|$ denote the maximum number of time labels assigned to any one edge in $(G,\lambda)$, and call this quantity \emph{temporality} following Mertzios et al. \cite{mertzios_temporal_2019}. The \emph{lifetime} of a temporal graph, denoted $T(G,\lambda)$, is the largest $t\in \mathbb{N}$ such that there exists an edge $e\in E(G)$ with $t\in\lambda(e)$.

\begin{figure}
    \centering
    \begin{tikzpicture}[%
    node distance=2cm,
    vert/.style={draw, circle, minimum size=2mm, fill, inner sep=0pt},
    edge/.style={draw, thick},
    every edge quotes/.style = {auto, font=\footnotesize}
    ]
        \node[vert,label={above:$a$}] (a) {};
        \node[vert,label={above:$b$}, right of=a] (b) {};
        \node[vert,label={above:$c$}, right of=b] (c) {};
        \node[vert,label={above:$d$}, right of=c] (d) {};
        \node[vert,label={above:$e$}, right of=d] (e) {};
        
        \draw[edge] (a) edge["{$1,3$}"'] (b);
        \draw[edge] (b) edge["{$2$}"'] (c);
         \draw[edge] (c) edge["{$2,3,4$}"'] (d);
        \draw[edge] (d) edge["{$4$}"'] (e);
        
    \end{tikzpicture}
    \caption{An example temporal graph $(G,\lambda)$ with maximum lifetime $4$ and temporality $3$. There is a temporal path from the vertex $a$ to all other vertices in the graph. The foremost temporal path from $a$ to $d$ arrives at time $3$.}
    \label{fig:tgraph-eg}
\end{figure}

To discuss reachability, we first need a notion of a temporal path.
Here we work with \emph{strict temporal paths}, where a \emph{strict temporal path} from $v_0$ to $v_{\ell}$ in $(G,\lambda)$ is a sequence of time-edges $((v_0v_1,t_1),(v_1v_2,t_2)\dots,(v_{\ell - 1}v_{\ell},t_{\ell}))$ where, for each $i$, $t_i \in \lambda(v_{i-1}v_i)$ and $t_1 < t_2 < \dots < t_{\ell}$.
In this paper, all temporal graphs discussed are strict, and so we will simply refer to them as temporal paths for brevity.
We say that the length of a temporal path $((v_0v_1,t_1),\dots,(v_{\ell - 1}v_{\ell},t_{\ell}))$ is $\ell$.
Given a temporal path $P\coloneqq ((v_0v_1,t_1),\dots,(v_{\ell - 1}v_{\ell},t_{\ell}))$, for each $1 \leq i < \ell$ the path $P'\coloneqq ((v_0v_1,t_1),\dots,(v_{i-1}v_i,t_i))$ is a \emph{prefix} of $P$.

If there exists a temporal path from $v_s$ to $v_t$ in $(G,\lambda)$, we say that $v_t$ is \emph{reachable} from $v_s$. For a given vertex $v_s$ in a temporal graph $(G,\lambda)$, we define $\reach((G,\lambda),v_s)$ to be the set of all vertices reachable from $v_s$ in $(G,\lambda)$. The \emph{maximum temporal reachability} of a graph $R_{\max}(G,\lambda)\coloneqq \max_{v\in V(G)}|\reach((G,\lambda),v)|$ is the cardinality of the largest reachability set in $(G,\lambda)$.

We say that the \emph{arrival time} of a path $P:= ((v_0v_1,t_1),\ldots, (v_{\ell-1}v_\ell,t_\ell))$ is $t_\ell$.
A path $P$ is \emph{foremost} from a set of paths $\mathcal{P}$ if there is no other path $P'\in \mathcal{P}$ with a strictly earlier arrival time.
We will write $\tfo((G,\lambda),v_s,v_t)$ to be the arrival time of a foremost temporal path from $v_s$ to $v_t$, where $t_{\text{Fo}}((G,\lambda),v_s,v_s) := 0$ and, if there is no temporal path from $v_s$ to $v_t$, we define $t_{\text{Fo}}((G,\lambda),v_s,v_t)\coloneqq \infty$.

Bui-Xuan et al.~\cite{XUAN2003} introduce a notion that we call a \emph{ubiquitous foremost path} (they refer to this as a ubiquitous foremost journey).
Let $u$ and $v$ be two vertices in a temporal graph $(G,\lambda)$. A foremost temporal path from $u$ to $v$ is a \emph{ubiquitous foremost path} if all of its prefixes are themselves foremost temporal paths.
%
%
Given a source vertex $v_s$ in a temporal graph $(G, \lambda)$, a \emph{ubiquitous foremost temporal path tree} on $v_s$ is defined to be a temporal tree $(G_T,\lambda_T)$ with the following properties:
\begin{itemize}
    \item $G_T$ is a subtree of $G$,
    \item for each $e \in E(G_T)$, $\lambda_T(e) \subseteq \lambda(e)$ with $|\lambda_T(e)| = 1$, and
    \item for each $v_t \in \reach((G,\lambda),v_s)$, $(G_T,\lambda_T)$ contains a ubiquitous foremost temporal path from $v_t$ to $v_s$ that arrives at the same time as a ubiquitous foremost temporal path from $v_t$ to $v_s$ in $(G,\lambda)$.
\end{itemize}
In particular, note that while there may be multiple ubiquitous foremost temporal paths from $v_s$ to $v_t$ in $(G,\lambda)$, $(G_T,\lambda_T)$ will contain exactly one ubiquitous foremost temporal path from $v_s$ to $v_t$, and the arrival time of a foremost temporal path from $v_s$ to any $v_t$ can be trivially inferred as the earliest active time over all temporal edges incident to $v_t$.  We note that a ubiquitous foremost temporal path tree can be computed efficiently.
\medskip
\begin{theorem}[Theorem 2 in~\cite{XUAN2003}]\label{theorem:ubiquitous-foremost-tree}
 Given a temporal graph $(G,\lambda)$ and a source vertex $v_s\in V(G)$, a ubiquitous foremost temporal path tree on $v_s$ can be calculated in $O(m(\log n + \log \tau(G,\lambda)))$ time. 
\end{theorem}
\medskip
By applying Theorem~\ref{theorem:ubiquitous-foremost-tree} for each source vertex $v_s$ in $V(G)$, we get the following corollary.
\medskip
\begin{corollary}\label{cor:temp-reach}
    Given a temporal graph $(G,\lambda)$ we can calculate $R_{\max}$ in $O(nm (\log n + \log \tau(G,\lambda)))$ time.
\end{corollary}
\medskip
Note that we can reduce the complexity of this to $O(n(n+m))$ by using Algorithm 1 of Wu et al.~\cite{Wu2014PathProblems} which calculates, for a given source $v_s$, foremost arrival times (but not actual temporal paths) for each vertex $v_t$ in $O(n+m)$ time.
This, however, does need the input data to be in a specified streaming format which allows their algorithm to calculate foremost arrival times in one pass.
For our purposes, knowing foremost arrival times is not as useful as knowing actual foremost temporal paths, so we omit further details and instead refer the reader to~\cite{Wu2014PathProblems}.

\subsection{Problems considered and initial observations}

Here we introduce the notion of perturbing the appearances of time-edges. To that end, we define $\delta$-perturbations as follows.
\medskip
\begin{definition}\label{dfn:perturbation}
     A temporal assignment $\lambda'$ is a $\delta$\emph{-perturbation} of $\lambda$ if there is a map $f$ from $\mathcal{E}(G,\lambda)$ to $E(G)\times [1, T(G,\lambda)+\delta]$ such that, for all time-edges, $f((e,t))=(e,t')$ where $t'$ is an integer in the interval $[\max(1,t-\delta), t+\delta]$. We define a new temporal assignment $\lambda'$ as $\lambda'(e)=\{t' \, \mid f(e,t)=(e,t')\}$.
    Here we say $\lambda$ has been \emph{perturbed} by $\delta$, or $\lambda'$ is a $\delta$-perturbation of $\lambda$. We also refer to any time-edge such that $(e,t)\neq f(e,t)$ as perturbed from $\lambda$ to $\lambda'$.
\end{definition}
\medskip


In some scenarios, there may be a bound on not only the amount of perturbations possible on the edges, but also the \emph{number} of time-edges perturbed. We capture this with the parameter $\zeta$.
\medskip
\begin{definition}\label{dfn:zeta-perturbation}
    A temporal assignment $\lambda'$ is a $(\delta,\zeta)$\emph{-perturbation} of $\lambda$ if it is a $\delta$-perturbation of $\lambda$ and 
    $\left|\{(e,t)\mid (e,t) \text{ perturbed from } \lambda \text{ to }\lambda'\right\}|\leq \zeta$. 
\end{definition}
\medskip
We can now define the main problem we consider in this paper.

\begin{framed}
\noindent
\textbf{\textsc{Temporal Reachability with Limited Perturbation (TRLP)}}\\
\textit{Input:} A temporal graph $(G,\lambda)$ and positive integers $\zeta$, $h\leq n$, and $\delta$.\\
\textit{Question:} Is there a $(\delta,\zeta)$-perturbation $\lambda'$ of $\lambda$ such that $R_{\max}(G,\lambda') \geq h$?
\end{framed}
\label{TR-PerturbLimit}

We also consider a straightforward special case of this problem, in which the number of perturbations is unrestricted.

\begin{framed}
\noindent
\textbf{\textsc{Temporal Reachability with Perturbation (TRP)}}\\
\textit{Input:} A temporal graph $(G,\lambda)$ and positive integers $h\leq n$ and $\delta$.\\
\textit{Question:} Is there a $\delta$-perturbation $\lambda'$ of $\lambda$ such that $R_{\max}(G, \lambda') \geq h$?
\end{framed}
\label{TR-Perturb}




Let $N_G(v)$ denote the \emph{neighbourhood} of a vertex, that is the set of all vertices which share an edge with $v$ in $G$.  We note a simple fact about neighbourhoods and reachability.
\medskip
\begin{obs}[Observation 3 in ~\cite{enright_assigning_2021}]\label{obs:neighbourhood}
    For any undirected temporal graph $(G,\lambda)$ and $v\in V(G)$, we have $\reach((G,\lambda),v)\supseteq N_G(v)\cup\{v\}$.
\end{obs}
\medskip
We note, as a result, that if $h\leq \Delta_G+1$ where $\Delta_G$ is the maximum degree of $G$, then $((G,\lambda),\zeta,h,\delta)$ is a yes-instance of TRLP since there must be a vertex in $(G,\lambda)$ with temporal reachability at least $\Delta_G+1$ regardless of perturbation.

We now show that if we only \emph{add} to the set of times at which edges are active, we cannot delay the arrival of a foremost temporal path, nor can we reduce temporal reachability.
\medskip
\begin{lemma}\label{lemma:foremost-arrival-time-subset}
    Let $(G,\lambda_1)$ and $(G, \lambda_2)$ be two temporal graphs on the same underlying graph $G$. If $\mathcal{E}((G,\lambda_1)) \subseteq \mathcal{E}((G,\lambda_2))$, then for each $v_s,v_t \in V(G)$, $t_{\text{Fo}}((G,\lambda_1),v_s,v_t) \geq t_{\text{Fo}}((G,\lambda_2),v_s,v_t)$.
\end{lemma}
\begin{proof}
    For any $v_s,v_t\in V(G)$, let $P$ be a foremost temporal path in $(G,\lambda_1)$ from $v_s$ to $v_t$. But as $\lambda_1(e) \subseteq \lambda_2(e)$ for any $e\in E(G)$, this path $P$ must also exist in $(G,\lambda_2)$.
\end{proof}

This gives us the following corollary.
\medskip
\begin{corollary}\label{cor:temp-reach-time-subset}
    Let $(G,\lambda_1)$ and $(G, \lambda_2)$ be two temporal graphs on the same underlying graph $G$. If $\mathcal{E}((G,\lambda_1)) \subseteq \mathcal{E}((G,\lambda_2))$, then for each $v \in V(G)$, $\reach((G,\lambda_1),v) \subseteq \reach((G,\lambda_2),v)$, and hence $R_{\max}(G,\lambda_1) \leq R_{\max}(G,\lambda_2)$.
\end{corollary}
\medskip

We now show that we can take a temporal graph with an arbitrary number of time labels on each edge and remove all but one of these time labels without affecting the temporal reachability of a given vertex $v_s$.
We do this by only keeping, for each edge, the earliest time label at which any path from $v_s$ might use the edge.
Any earlier times by definition cannot be used by any path, and so are not necessary, and for any path that use the edge at a later time there is some path that reaches the same vertices, but uses the edge at the earlier time.
\medskip
\begin{lemma}\label{lemma:temp-cannot-affect-reach-by-removing-later-times}
   Given a temporal graph $(G,\lambda)$ and vertex $v\in V(G)$ with $\reach_{G,\lambda}(v) = R$, let $\mathcal{P}$ be a set of ubiquitous foremost temporal paths such that for each $u\in R$, $\mathcal{P}$ contains a ubiquitous foremost temporal path from $v$ to $u$,
   and let $\lambda'$ be a temporal labelling function such that
   for each $e\in E(G)$, if $e$ appears in at least one path in $\mathcal{P}$ then $\lambda'(e) = \{ \min \{ t \mid (e,t) \in P \text{ for some } P\in\mathcal{P}\}\}$.
   Then for any $u\in R$, $t_{\text{Fo}}((G,\lambda),v,u) = t_{\text{Fo}}((G,\lambda'),v,u)$, and $\reach((G,\lambda),v) = \reach((G,\lambda'),v)$. 
\end{lemma}
\begin{proof}
   First, note that as $\mathcal{E}((G, \lambda')) \subseteq \mathcal{E}((G, \lambda))$ we get $\reach_{G,\lambda'}(v) \subseteq \reach_{G,\lambda}(v)$ from \Cref{cor:temp-reach-time-subset}, so it remains to show that $\reach_{G,\lambda}(v) \subseteq \reach_{G,\lambda'}(v)$.
   Towards this, consider an arbitrary vertex $u\in R$. We will show that $u$ must also be an element of $\reach_{G,\lambda'}(v)$; thus $\reach_{G,\lambda}(v) \subseteq \reach_{G,\lambda'}(v)$. Let $P$ be the ubiquitous foremost temporal path from $v$ to $u$ in $(G,\lambda)$. Then, for each time-edge $(e,t)$ in $P$, $(e,t)$ must be in $\mathcal{E}((G,\lambda'))$. Thus, $u$ is temporally reachable from $v$ in $(G,\lambda')$ and the foremost temporal path from $v$ to $u$ arrives at the same time in both $(G,\lambda)$ and $(G,\lambda)$, as required.
\end{proof}

\section{Intractability of TRLP}\label{sec:intractability}
We begin by analysing the complexity of TRLP.  We do this by giving a polynomial-time reduction from \textsc{Dominating Set}~\cite{subconstant-pcp,Cygan2015} to TRLP, which lets us show that the problem
\begin{itemize}
    \item is \NP-complete,
    \item is \W[2]-complete with respect to $\zeta$,
    \item is \NP-hard to approximate to within a factor of $\alpha \log n$ for some $\alpha > 0$, and
    \item has no $f(\zeta)n^{o(\zeta)}$-time algorithm unless ETH fails. 
\end{itemize}
Following this, we show that it is unlikely that the ETH lower bound can be significantly improved, as we give an algorithm running in time $n^{O(\zeta)}\log \tau(G,\lambda)$. 

We begin by describing a particular construction of a temporal graph from a static graph.
\medskip
\newcommand{\consds}[1]{\ensuremath{\mathfrak{C}#1}}
\begin{const}\label{const:domset}
Given a graph $G$, let $\consds{(G)} = (G',\lambda)$ be the temporal graph constructed as follows.
Create a vertex $v_s$ in $V(G')$, and for each vertex $v_i \in V(G)$ create two vertices $v^1_i$ and $v^2_i$ in $V(G')$.
For each $v_i \in V(G)$, add the edges $v_sv^1_i$ and $v^1_iv^2_i$ to $(G')$.
Then, for each edge $v_iv_j\in E(G)$, add the edges $v^1_i v^2_j$ and $v^1_jv^2_i$ to $(G')$.
Let $\lambda(e) = 2$ for each $e\in E(G')$.
A sample construction is given in \Cref{fig:trlp-no-fpt}.
\end{const}

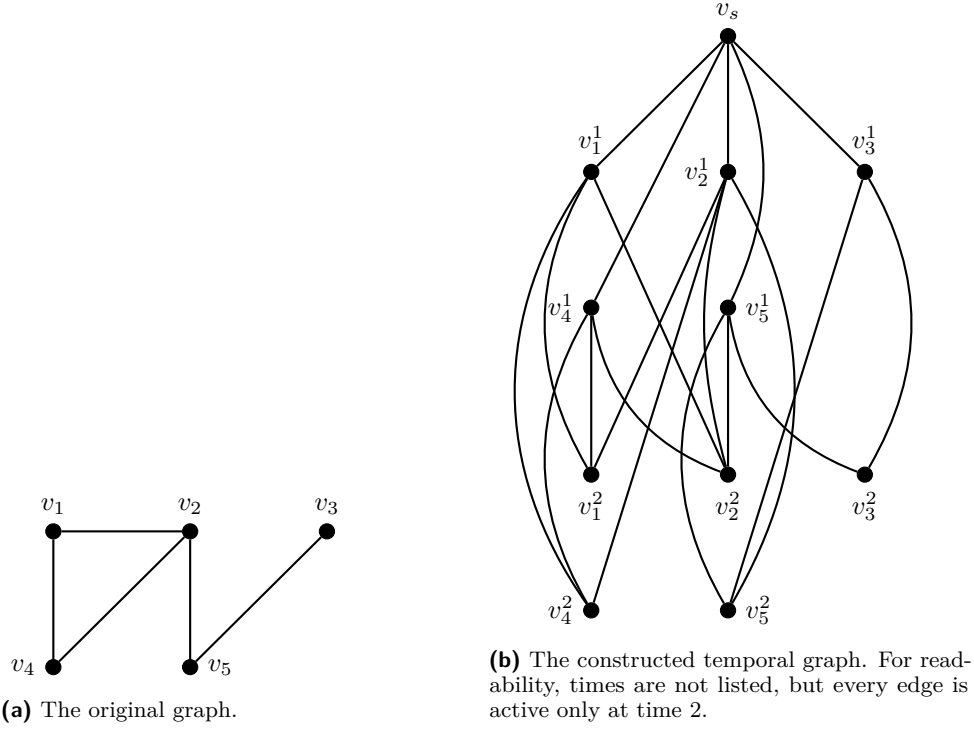
\begin{figure}[ht]
    \centering
    \begin{subfigure}{0.45\textwidth}
    \begin{tikzpicture}[%
    node distance=18mm,
    vert/.style={draw, circle, minimum size=2mm, fill, inner sep=0pt},
    edge/.style={draw, thick},
    every edge quotes/.style = {auto, font=\footnotesize}
    ]
        \node[vert,label={above:$v_1$}] (v1) {};
        \node[vert,label={above:$v_2$}, right of=v1] (v2) {};
        \node[vert,label={above:$v_3$}, right of=v2] (v3) {};
        \node[vert,label={left:$v_4$}, below of=v1] (v4) {};
        \node[vert,label={right:$v_5$}, right of=v4] (v5) {};

        \draw[edge] (v1) -- (v2);
        \draw[edge] (v1) -- (v4);
        \draw[edge] (v2) -- (v4);
        \draw[edge] (v2) -- (v5);
        \draw[edge] (v3) -- (v5);
    \end{tikzpicture}
    \caption{The original graph.}\label{fig:trlp-no-fpt-ds}
    \end{subfigure}
    \begin{subfigure}{0.45\textwidth}
    \begin{tikzpicture}[%
    node distance=18mm,
    vert/.style={draw, circle, minimum size=2mm, fill, inner sep=0pt},
    edge/.style={draw, thick},
    every edge quotes/.style = {auto, font=\footnotesize}
    ]
        \node[vert,label={above:$v_s$}] (vsa) {};
        \node[vert,label={left:$v^1_2$}, below of=vsa] (v2a) {};
        \node[vert,label={above:$v^1_1$}, left of=v2a] (v1a) {};
        \node[vert,label={above:$v^1_3$}, right of=v2a] (v3a) {};
        \node[vert,label={left:$v^1_4$}, below of=v1a] (v4a) {};
        \node[vert,label={right:$v^1_5$}, right of=v4a] (v5a) {};

        \draw[edge] (vsa) edge (v1a);
        \draw[edge] (vsa) edge (v2a);
        \draw[edge] (vsa) edge (v3a);
        \draw[edge] (vsa) edge (v4a);
        \draw[edge] (vsa) edge[bend left=25] (v5a);
        
        \node[vert,label={below:$v^2_2$}, below=20mm of v5a] (v2b) {};
        \node[vert,label={below:$v^2_1$}, left of=v2b] (v1b) {};
        \node[vert,label={below:$v^2_3$}, right of=v2b] (v3b) {};
        \node[vert,label={left:$v^2_4$}, below of=v1b] (v4b) {};
        \node[vert,label={right:$v^2_5$}, right of=v4b] (v5b) {};

        \draw[edge] (v1a) edge[bend right=30] (v1b);
        \draw[edge] (v1a) -- (v2b);
        \draw[edge] (v1a) edge[bend right=35] (v4b);
        
        \draw[edge] (v2a) edge (v1b);
        \draw[edge] (v2a) edge[bend right=15] (v2b);
        \draw[edge] (v2a) edge (v4b);
        \draw[edge] (v2a) edge[bend left=30] (v5b);
        
        \draw[edge] (v3a) edge[bend left=30] (v3b);
        \draw[edge] (v3a) edge (v5b);
        
        \draw[edge] (v4a) edge[bend right=30] (v4b);
        \draw[edge] (v4a) edge (v1b);
        \draw[edge] (v4a) edge[bend right=30] (v2b);
        
        \draw[edge] (v5a) edge[bend right=30] (v5b);
        \draw[edge] (v5a) edge (v2b);
        \draw[edge] (v5a) edge[bend right=30] (v3b);

    \end{tikzpicture}
    \caption{The constructed temporal graph. For readability, times are not listed, but every edge is active only at time 2.}\label{fig:trlp-no-fpt-temporal}
    \end{subfigure}
    \caption{Diagram of the construction from Construction~\ref{const:domset}.}\label{fig:trlp-no-fpt}
\end{figure}

Note that $G'$ is a bipartite graph where the sets $\{v_s\} \cup \{v^2_a \mid v_a\in V(G)\}$ and $\{v^1_a \mid v_a \in V(G)\}$ form the two partitions.
\medskip
\begin{obs}
    Given a graph $G$, $\consds{(G)}$ can be constructed in time polynomial in the size of $G$.
\end{obs}
\medskip
We now demonstrate the relationship between the existence of a dominating set in $G$ and the existence of a solution to TRLP in $\consds{(G)}$.

\medskip
\begin{lemma}\label{lemma:ds-to-reachability}
    Given an instance $(G,r)$ of \textsc{Dominating Set} with $r\geq 1$, $G$ has a dominating set of size $r$ if and only if $((G',\lambda),1,r,|V(G')|)$ is a yes-instance of \trlpp, where $(G',\lambda) = \consds{(G)}$.
\end{lemma}
\begin{proof}
    We begin by showing that if we have a dominating set $\mathcal{C} =\{v_1,\ldots, v_r\}$ of $G$, we can find a $(1,r)$-perturbation $\lambda'$ of $\lambda$ such that $v_s$ has temporal reachability $|V(G')|$.
    We create $\lambda'$ by perturbing the edges $v_sv^1_i$ for $v_i\in \mathcal{C}$ to be active at time 1.
    Clearly $\lambda'$ is a $(1,r)$-perturbation of $\lambda$.
    We will show that $v_s$ can reach every other vertex in $V(G')$.
    Let $v$ be an arbitrary vertex in $V(G')$.
    If $v=v_s$ then clearly $v_s$ reaches $v$.
    If $v = v^1_a$ for some $a$, then $v_s$ can reach $v$ along a temporal path of one edge.
    If $v = v^2_a$ for some $a$, then let $v_b$ be the vertex in $\mathcal{C}$ that dominates $v_a$ in $G$.
    Then $\lambda'(v_sv^1_b) = 1$ and $\lambda'(v^1_bv^2_a)=2$, so $v_s$ can reach $v$ along a temporal path of two edges.
    As $v$ was taken arbitrarily, there is a temporal path in $(G',\lambda')$ from $v_s$ to every other $v\in V(G')$, so $((G',\lambda),1,r,|V(G')|)$ is a yes-instance of \trlpp.

    We now show that if $((G',\lambda),1,r,|V(G')|)$ is a yes-instance of \trlpp, there exists a dominating set of cardinality $r$. Assume that $((G',\lambda),1,r,|V(G')|)$ is a yes-instance of \trlpp.
    Let $v$ be a vertex in $V(G')$ and let $\lambda'$ be a $(1,r)$-perturbation of $\lambda$ such that $v$ that has temporal reachability $|V(G')|$ in $(G',\lambda')$.
    Note that for any $(1,r)$-perturbation $\lambda'$ of $\lambda$ and for any edge $e\in E(G')$, $\lambda'(e) \leq 3$.
    As a result, a temporal path in $(G',\lambda')$ can have a length of at most three.


    We now take cases based on whether the vertex $v\in V(G')$ is $v_s$, is of the form $v^1_a$ for some $a$, or is of the form $v^2_a$ for some $a$.
    If $v=v_s$, we show how to take the perturbation $\lambda'$ and construct a dominating set.
    For the other cases we show either how to construct a dominating set, or how to construct a $(1,r)$-perturbation such that $v_s$ can reach all other vertices, and then utilise the first case to complete the argument.
    
    \textbf{Case 1: $v=v_s$}
    Let $C = \{v_a \mid \lambda(v_sv^1_a) = 1 \vee \exists \; v^2_b : \lambda(v^1_av^2_b) = 3\}$.
    That is, $C$ contains the vertex $v_a$ if and only if at least one edge incident to $v^1_a$ is perturbed to create $\lambda'$.
    It follows that $|C| \leq r$; we now show that $C$ is a dominating set in $G$.
    Let $v_b$ be an arbitrary vertex in $G$, and let $v_sv^1_av^2_b$ be a temporal path in $(G',\lambda')$.
    Such a path must exist as $v$ can temporally reach all vertices in $(G',\lambda')$, and it must have a length of two as there is no path of length either one or three in $G'$ from $v_s$ to $v^2_b$, and there are no temporal paths of length $\geq 4$ in $(G',\lambda')$.
    It must also be that $\lambda'(v_sv^1_a) < \lambda'(v^1_av^2_b)$, and since $\lambda(v_sv^1_a) = \lambda(v^1_av^2_b) = 2$, it follows that at least one of $\lambda'(v_sv^1_a) = 1$ or $\lambda'(v^1_av^2_b) = 3$ must hold.
    However, then $v_a\in C$, and as $v^1_av^2_b \in E(G')$, by construction we also have $v_av_b\in E(G)$ and so $v_a$ dominates $v_b$.
    As $v_b$ was taken arbitrarily, $C$ is a dominating set of size at most $r$ in $G$.

    \textbf{Case 2: $v=v^1_a$ for some $a$:}
    First, we will show that if there is no edge $e$ incident to $v^1_a$ such that $\lambda'(e) = 1$, then $\{v_a\}$ forms a dominating set.
    If no such edge is perturbed, then all temporal paths from $v^1_a$ must have length at most two.
    As $G'$ is bipartite, this means that for $v^1_a$ to be able to temporally reach any vertex of the form $v^2_b$, $v^1_a$ must be adjacent to it (and thus by construction $v_a$ is adjacent to every other vertex $v_b$ in $G$).
    As $r\geq 1$, $\{v_a\}$ is a dominating set in $G$ with size at most $r$.

    We next assume that there is some edge $e$ incident to $v^1_a$ with $\lambda'(e) = 1$.
    We will construct a $(1,r)$-perturbation $\lambda''$ of $\lambda$ such that $v_s$ can temporally reach all vertices in $(G',\lambda'')$;
    the result then follows by applying Case 1.
    We achieve this by ensuring that for each perturbed edge in $\lambda'$, we select at most one edge in $\lambda''$ to perturb.
    Recalling that at least one edge incident to $v^1_a$ is perturbed, let $\lambda''(v_s v^1_a) = 1$, and let $\lambda''(v^1_a v^2_b) = 2$ for every edge of the form $v^1_a v^2_b$ in $E(G')$.
    For any vertex $v^1_b$ such that there exists a vertex $v^2_c$ with $\lambda'(v^1_bv^2_c) = 3$, let $\lambda''(v_sv^1_b) = 1$, and 
    for every other edge $e$, let $\lambda''(e) = 2$.
    By this construction, $\lambda''$ is a $(1,r)$-perturbation of $\lambda$ as $\lambda'$ is a $(1,r)$-perturbation of $\lambda$.
    Note that $v_s$ is adjacent to, and thus can temporally reach, every vertex of the form $v^1_b$ for any $b$.
    Consider next an arbitrary vertex of the form $v^2_b$.
    If $v^2_b$ is adjacent to $v^1_a$, then $v_s$ can temporally reach $v^2_b$ along the temporal path $v_sv^1_av^2_b$ in $(G,\lambda'')$.
    Otherwise, $v^2_b$ is not adjacent to $v^1_a$.
    As $v^1_a$ can temporally reach $v^2_b$ in $(G',\lambda')$, temporal paths in $(G',\lambda')$ have length at most 3, and $G'$ is bipartite, there is a vertex $y$ and a temporal path $v^1_ayv^1_cv^2_b$ in $(G',\lambda')$, where either $y=v_s$ or $y$ is of the form $v^2_d$.
    This means that $\lambda'(v^1_cv^2_b) = 3$, so $\lambda''(v_sv^1_c) = 1$ and $v_sv^1_cv^2_b$ is a temporal path in $(G',\lambda'')$.
    Thus $v_s$ can temporally reach every vertex in $G'$ in $(G',\lambda'')$,
    and so by applying Case 1 there is a dominating set $C$ of $G$ of size at most $r$.
    
    \textbf{Case 3: $v=v^2_a$ for some $a$:}
    We will construct a $(1,r)$-perturbation $\lambda''$ of $\lambda$ such that $v_s$ can temporally reach all vertices in $(G',\lambda'')$; 
    the result then follows by applying Case 1.
    We achieve this by ensuring that for each perturbed edge in $\lambda'$, we select at most one edge in $\lambda''$ to perturb.
    For every vertex $v^1_b$ incident to at least one edge $e$ with $\lambda(e) \neq \lambda'(e)$, let $\lambda''(v_sv^1_b) = 1$,
    and for every vertex $v^1_b$ with $\lambda'(v_sv^1_b)=3$, let $\lambda''(v_sv^1_b) = 1$.
    For any vertex of the form $v^1_b$ incident to one or more edges $v^1_bv^2_c$ with $\lambda'(v^1_bv^2_c) = 3$, let $\lambda''(v^sv^1_b) = 1$.
    For every other edge $e$, let $\lambda''(e) = 2$.
    By construction, $\lambda''$ is a $(1,r)$-perturbation of $\lambda$ as $\lambda'$ is a $(1,r)$-perturbation of $\lambda$.
    Note that $v_s$ is adjacent to, and thus can temporally reach, every vertex of the form $v^1_b$ for any $b$.
    Consider an arbitrary vertex of the form $v^2_b$.
    Recall that $G'$ is bipartite, and that one of the partitions is given by $\{v_s\}\cup \{v^2_b \mid v_b \in G \}$.
    As temporal paths in any $(1,r)$-perturbation of $\lambda$ can have maximum length 3, we therefore know that any temporal path from $v^2_a$ to $v^2_b$ in $(G',\lambda'')$ must have length either zero (if $v^2_a = v^2_b$) or two, and that any temporal path from $v^s$ to $v^2_b$ in $(G',\lambda'')$ must have length two.
    If $v^2_b = v^2_a$, then as $v^2_a$ can temporally reach $v_s$ along some temporal path $v^2_av^1_cv_s$ for some $v^1_c$, where at least one of $\lambda'(v^2_av^1_c) \neq 2$ or $\lambda'(v^1_cv_s)\neq 2$ must hold, we have $\lambda''(v_sv^1_c) = 1$ and so $v_s$ can reach $v^2_b$ along the temporal path $v_sv^1_cv^2_b$.
    Else $v^2_b \neq v^2_a$, and so there is a temporal path $v^2_av^1_cv^2_b$ in $(G',\lambda')$ for some vertex $v^1_c$.
    It must be that $\lambda'(v^2_av^1_c) < \lambda'(v^1_cv^2_b)$, so either $\lambda'(v^2_av^1_c) \neq \lambda(v^2_av^1_c)$ or $\lambda'(v^1_cv^2_b) \neq \lambda(v^1_cv^2_b)$ holds.
    This means that $\lambda''(v_sv^1_c) = 1$ and $\lambda''(v^1_cv^2_b)=2$ so $v_s$ can temporally reach $v^2_b$.
    Thus $v_s$ can temporally reach every vertex in $G'$ in $(G',\lambda'')$,
    and so by applying Case 1 there is a dominating set $C$ of $G$ of size at most $r$.
\end{proof}

\newcommand{\optp}[1]{\ensuremath{\text{OPT}_\zeta(\text{#1})}}

We can now use this equivalence to give a number of hardness and lower-bound results. Together with the observation that the perturbation $\lambda'$ serves as a certificate for the problem, Lemma~\ref{lemma:ds-to-reachability} gives us \NP-completeness of TRLP.  We further observe that, given an instance of TRLP with all parameters except $\zeta$ specified, approximating the smallest value of $\zeta$ such that the resulting instance is a yes-instance to within a logarithmic factor is \NP-hard; this follows as the size of a smallest dominating set cannot be approximated to a factor better than $\alpha \log n$ for some $\alpha > 0$ (see \cite{subconstant-pcp}).
 In addition, \textsc{Dominating Set} is \W[2]-hard parameterised by the size of the dominating set (see e.g.~\cite[Proposition 13.20]{Cygan2015}). Therefore, TRLP must be \W[2]-hard parameterised by $\zeta$. Finally, assuming ETH, there is no $f(k)n^{o(k)}$-time algorithm for \textsc{Dominating Set} (see \cite[Corollary 14.23]{Cygan2015}). This leads us to the following corollary.
\medskip
\begin{theorem}\label{thm:inapprox-reach}
    TRLP is \NP-complete. Moreover:
    \begin{itemize}[leftmargin=0.5cm]
        \item For some $\alpha>0$ it is \NP-hard to approximate to within a factor $\alpha \log n$ the minimum number $\zeta$ of $\delta$-perturbations required to transform the input temporal graph $(G,\lambda)$ into a temporal graph with maximum reachability at least $h$.
        \item TRLP is \W[2]-hard parameterised by $\zeta$, the number of edges that may be perturbed.
        \item Assuming ETH, there is no $f(\zeta)n^{o(\zeta)}$-time algorithm for TRLP for any computable function $f$.
    \end{itemize}
\end{theorem}
\medskip



We now give an algorithm which shows that our ETH lower bound is close to best possible. The algorithm considers every possible subset of edges of size $\zeta$, and then for each such subset performs a Dijkstra-like exploration of a temporal graph using any edge in the selected subset as early as possible.
We begin by giving \Cref{alg:trlp-xp-inner}, which performs this exploration for a given source and given set of perturbed edges in $O(n^2)$ time. Repeatedly calling this algorithm for each possible set of perturbed edges and each possible source gives us the desired result.

For \Cref{alg:trlp-xp-inner}, we need the concept of a \emph{priority queue}.
A priority queue is a data structure for storing data with an associated priority.
For our purposes, priority queues allow for insertion of data as well as the removal of one entry with highest priority.
Assuming the priority queue stores at most $c$ elements, priority queues can be implemented to allow for either constant time insertion and $O(c)$ removal, $O(c)$ insertion and constant time removal, or $O(\log c)$ insertion and $O(\log c)$ removal.
Any of these suffice for our result, and for further information on priority queues we point the reader to~\cite[Section 3.5]{algorithmdesign}.
\begin{algorithm}
\caption{An algorithm which determines if there exists a perturbation of the input temporal graph such that a source vertex reaches at least $k$ vertices, given the exact edges which may be perturbed}\label{alg:trlp-xp-inner}
\begin{algorithmic}[1]
\Require A temporal graph $(G,\lambda)$, a source vertex $v_s$, non-negative integers $k$ and $\delta$, and a set $E'$ of edges of $G$.
\Ensure True if and only if there is a $\delta$-perturbation $\lambda'$ of $\lambda$ where $\lambda'(e) \neq \lambda(e)$ implies $e\in E'$ for all $e\in E(G)$ and $\text{reach}((G,\lambda'),v_s)\geq k$.
\State Let $Q$ be a time-based priority queue with priority for earlier times. We build $Q$ such that if $(v_1,v_2,t)$ is in $Q$ then there exists a temporal path that reaches $v_2$ at time $t$. 
\State Let $\mathcal{V}$ be a set such that if vertex $v$ is in $\mathcal{V}$ then a temporal path from $v_s$ to $v$ exists.
\State Add $v_s$ to $\mathcal{V}$ \label{alg:trlp-xp-base}
\ForAll{edges $(v_s,v_2)$ incident to $v_s$}\label{alg:trlp-xp:edges_first}
    \State Add $v_2$ to $\mathcal{V}$
    \If{ $(v_s,v_2)\in E'$}
        \State Let $t' = \min(\{ t + d \mid d \in [-\delta,\delta] \wedge t+d\geq 1 \wedge t\in \lambda(v_s,v_2) \})$\label{alg:trlp-xp:tmin_first}
        \State Add $(v_s,v_2,t')$ to $Q$\label{alg:trlp-xp:add_first_p}
    \Else
        \State Add $(v_s,v_2,\min(\lambda(v_s,v_2)))$ to $Q$\label{alg:trlp-xp:add_first_no_p}
    \EndIf 
\EndFor
\While{$Q$ is not empty}
    \State Let $(v_1, v_2,t)$ be a tuple of highest priority removed from $Q$\label{alg:trlp-xp:off_q}
    \If{$v_2\not\in \mathcal{V}$}
        \State Add $v_2$ to $\mathcal{V}$
        \ForAll{edges $(v_2,v_3)$ incident to $v_2$ where $v_3\notin \mathcal{V}$}\label{alg:trlp-xp:edges}
            \If{ $v_2v_3\in E'$}
                \If{$\max\lambda(v_2,v_3) + \delta > t$}
                    \State Set $t' = \min(\{ t^* + d \mid d \in [-\delta,\delta] \wedge t^*\in \lambda(v_2,v_3) \wedge t^* + d > t\})$\label{alg:trlp-xp:tmin}
                    \State Add $(v_2, v_3,t')$ to $Q$\label{alg:trlp-xp:add_p}
                \EndIf
            \Else
                \If{$\max(\lambda(v_2,v_3))> t$}
                    \State Set $t''=\min(\{\hat{t}\in \lambda(v_2,v_3)\,:\, \hat{t}>t\})$
                    \State Add $(v_2, v_3,t'')$ to $Q$\label{alg:trlp-xp:add_no_p}
                \EndIf
            \EndIf 
        \EndFor
    \EndIf
\EndWhile
\If {$|\mathcal{V}| \geq k$}
    \State \textbf{return} True
\EndIf
\State \textbf{return} False
\end{algorithmic}
\end{algorithm}
\medskip
\begin{theorem}\label{thm:trlp-xp}
    TRLP is solvable in $O\left(n^{2\zeta+3}\log(\tau(G,\lambda))\right)$ time.
\end{theorem}
\begin{proof}
    We solve TRLP by running Algorithm~\ref{alg:trlp-xp-inner} for each set of perturbed edges and each possible source. This requires that we run Algorithm~\ref{alg:trlp-xp-inner} $O(n\cdot \binom{n^2}{\zeta})$ times. 
    \begin{claim}\label{lem:trlp-xp-correct}
    Algorithm~\ref{alg:trlp-xp-inner} returns true if and only if there exists a perturbation of the edges in $E'$ such that a source vertex reaches at least $k$ vertices.
\end{claim}
\begin{proof}
    
    We show that a vertex is added to the set $\mathcal{V}$ if and only if there is a $\delta$-perturbation $\lambda'$ where that vertex is reachable from $v_s$ and a time-edge of $\lambda$ is not a time-edge of $\lambda'$ if and only if the edge is in $E'$. We can show this by induction on the time of arrival of a temporal path from the source vertex $v_s$. The base case is the source vertex itself, which is trivially temporally reachable from itself at time 0 and added to $\mathcal{V}$ in line~\ref{alg:trlp-xp-base} of the algorithm.

    Assume that, for all vertices temporally reachable from $v_s$ by time $t$, they are added to $\mathcal{V}$ if and only if they are reachable from $v_s$ under some perturbation of the edges in $E'$. Then, for any vertex $v_2$ not in $\mathcal{V}$ which is incident to a vertex $v_1$ in $\mathcal{V}$, $v_2$ is added to $\mathcal{V}$ if and only if $(v_1,v_2,t')$ is in $Q$ for some $t'>t$ and $v_2$ is not in $\mathcal{V}$. This occurs either in line~\ref{alg:trlp-xp:add_p} or line~\ref{alg:trlp-xp:add_no_p} (lines~\ref{alg:trlp-xp:add_first_p} or~\ref{alg:trlp-xp:add_first_no_p} if $v_1=v_s$) of the algorithm. In the former, the edge $v_1v_2$ is in $E'$ and there exists a $d\in [-\delta,\delta]$ and an occurrence of $v_1v_2$ at time $t^*$ such that $t<t^*+\delta=t'$. Since $v_1$ is reached by time $t$, the time-edge $(v_1v_2,t')$ can be appended to the walk from $v_s$ to $v_1$ to give a temporal walk from $v_s$ to $v_2$ arriving at time $t'$. Similarly, if the edge $v_1v_2$ is not in $E'$, the time-edge $(v_1,v_2,t')$ is in $Q$ if and only if there is a path to $v_2$ by time $t''<t'$ by the inductive hypothesis. In this case, the edge $v_1v_2$ occurs at time $t'$, therefore we can concatenate the path to $v_1$ from $v_s$ with the time-edge $(v_1v_2,t')$. Thus, a time-edge appears in $Q$ if and only if there is a walk from $v_s$ to both endpoints. Therefore, a vertex is added to $\mathcal{V}$ if and only if there exists a temporal path from $v_s$ to it by time $t$ for all $0\leq t\leq T(G,\lambda)$.

    Once $Q$ is empty, all vertices reachable from $v_s$ under $\lambda'$ must have been added to $\mathcal{V}$. Therefore, the algorithm returns true if and only if $|\mathcal{V}|\geq k$. This occurs if and only if there exists a perturbation of the edges in $E'$ such that a source vertex reaches at least $k$ vertices.
\end{proof}
\begin{claim}\label{lem:trlp-xp-time}
    Algorithm~\ref{alg:trlp-xp-inner} terminates in $O(n^2\log(\tau(G,\lambda)))$ time.
\end{claim}
\begin{proof}
    Using the data structure given by Bui-Xuan et al. \cite{XUAN2003}, we can look up the times at which an edge is active in $O(\log(\tau(G,\lambda)))$ time where $\tau$ is the temporality of the graph. Note that, in Algorithm~\ref{alg:trlp-xp-inner}, a vertex is added to $\mathcal{V}$ at most once and any edge $(u,v)$ is added to $Q$ as part of a triple $(u,v,t)$ at most twice (i.e. there are at most two times with which any edge appears in $Q$). Therefore, the algorithm runs in $O(n^2\log(\tau(G,\lambda)))$ time.
\end{proof}
    
    This gives us the desired result.
\end{proof}

\section{Tractability with many perturbations}\label{sec:many-perturbations}

In this section we show that allowing many perturbations makes our problem easier to solve. 
In particular, we prove that TRP is solvable in polynomial time, and also that TRLP is solvable in polynomial time if the number $\zeta$ of allowed perturbations is at least the target reachability.

As a warm-up, we begin by considering the case in which each perturbation may also be very large: suppose $\delta$ is at least the maximum of the lifetime of $(G,\lambda)$ and the diameter of $G$.  We claim that in this case, if $h$ is the target reachability and $\zeta \ge h -1$, we have a yes-instance.  To see this, fix an arbitrary subtree $H$ of $G$ with exactly $h$ vertices (noting that such a subtree contains $h - 1 \le \zeta$ edges, and if no such tree exists we have a no-instance of TRP), and further pick an arbitrary vertex $r \in V(H)$ which we shall call the root.  The permitted perturbations are large enough that we can perturb one appearance of each edge in $H$ so that, for each leaf $\ell$ of $H$, the times of these edge appearances are strictly increasing along the path from $r$ to $\ell$.  It follows that, in the perturbed graph, $r$ has reachability at least $h$, as required.






The first step towards our tractability results is to show that, given a source vertex in a temporal graph, and an integer $\delta$, we can identify a single $\delta$-perturbation that minimises foremost arrival times to all target vertices: one perturbation achieves such a result exactly because a ubiquitous foremost temporal path tree is a tree with edges active at exactly one time step. Recall that we denote the arrival time of a foremost temporal path from $v_s$ to $v_t$ in $(G,\lambda)$ by $t_{\text{Fo}}((G,\lambda),v_s,v_t)$.
\medskip
\begin{lemma}\label{lemma:temp-foremost-perturbed}
    Given a temporal graph $(G,\lambda)$, vertex $v_s\in V(G)$, and integer $\delta$, we can identify in $O(m(\log n + \log \delta \tau(G,\lambda)))$ time a $\delta$-perturbation $\lambda'$ such that for each $v_t\in V(G)$,
    $t_{\text{Fo}}((G,\lambda'),v_s,v_t) \leq t_{\text{Fo}}((G,\lambda''),v_s,v_t)$ for all $\delta$-perturbations $\lambda''$.
\end{lemma}
\begin{proof}
    First we create the new time-labelling function $\lambda''$ as follows.
    For each $e\in E(G)$, let $\lambda''(e) = \{ t + p \mid p \in [-\delta, \delta], t \in \lambda(e)\}$, noting that as a result $\tau(G,\lambda'') \leq (2\delta+1) \tau(G,\lambda) = O(\delta \tau(G,\lambda))$.
    Then, by Theorem~\ref{theorem:ubiquitous-foremost-tree} we can calculate a ubiquitous foremost temporal path tree $(G_T,\lambda_T)$ in $(G,\lambda'')$ in $O(m(\log n + \log \tau(G,\lambda''))) = O(m(\log n + \log \tau(G,\lambda)\delta))$ time.
    Then we create $\lambda'$ as follows.
    For each $e\in E(G_T)$, let $\lambda'(e) = \lambda_T(e)$, and for any $e \in E(G)\setminus E(G_T)$, let $\lambda'(e) = \lambda(e)$.

    We now show that $\lambda'$ satisfies the requirements of the theorem.
    By construction, $\mathcal{E}((G,\lambda_T)) \subseteq \mathcal{E}((G,\lambda''))$, and $\lambda'$ is a $\delta$-perturbation of $\lambda$.
    Take an arbitrary $v_t\in V(G)$.
    By definition of a ubiquitous foremost temporal path tree, there is some ubiquitous foremost temporal path $P$ in $(G_T,\lambda_T)$, and by construction, this path $P$ must also be a ubiquitous foremost temporal path in $(G,\lambda'')$.
    Thus the arrival time $t$ of $P$ must satisfy $t_{\text{Fo}}((G,\lambda''),v_s,v_t) = t$, and so as any $\delta$-perturbation $\lambda^\dagger$ of $\lambda$ satisfies $\mathcal{E}((G, \lambda^\dagger)) \subseteq \mathcal{E}((G, \lambda''))$, by Corollary~\ref{cor:temp-reach-time-subset} we get $t\leq t_{\text{Fo}}((G,\lambda^\dagger),v_s,v_t)$ for any $\delta$-perturbation $\lambda^\dagger$ of $\lambda$.
%
\end{proof}

\begin{corollary}
    Given a temporal graph $(G,\lambda)$, vertex $v_s\in V(G)$, and integer $\delta$, we can identify in $O(m(\log n + \log \delta \tau(G,\lambda)))$ time a $\delta$-perturbation $\lambda'$ of $\lambda$ such that for each $v_t\in V(G)$,
    $t_{\text{Fo}}((G,\lambda'),v_s,v_t) \leq t_{\text{Fo}}((G,\lambda^*),v_s,v_t)$ for all $\lambda^*$ which are $\delta$-perturbations of $\lambda$.
\end{corollary}
\medskip
We are now ready to solve \trppp using this result.
Recall that $\tau(G,\lambda) = \max_{e\in E(G)} |\lambda(e)|$.
\medskip
\begin{theorem}\label{thm:temp-reach-perturbed}
    We can solve \trppp in $O(nm (\log n + \log \delta \tau(G,\lambda)))$ time.
\end{theorem}
\begin{proof}
For each source vertex $v_s\in V(G)$, we use Lemma~\ref{lemma:temp-foremost-perturbed} to identify a $\delta$-perturbation $\lambda'$ that minimises the foremost arrival time to any target vertex $v_t\in V(G)$, and calculate  $\reach((G,\lambda'),v_s) = \{v_t\in V(G) \mid t_{\text{Fo}}((G,\lambda'),v_s,v_t) < \infty\}$.
If, for any source vertex $v_s$ we get $|\reach((G,\lambda'),v_s)| \geq h$ then the instance is a yes-instance to TRP, and additionally $\lambda'$ can act as a certificate to this fact.
\end{proof}

Now turning our attention to TRLP, we start with some further notation.  To consider a perturbation of fewest time-edges such that the maximum reachability is at least $h$, we introduce the following notion of worthwhile perturbations.
\medskip
\begin{definition}
    A perturbation $\lambda'$ of $\lambda$ is \emph{worthwhile} if reverting any time-edge that has been perturbed from $\lambda$ to $\lambda'$ back to the time assigned by $\lambda$ decreases the maximum reachability of the temporal graph. We call any time-edge $(e,t)$ of $(G,\lambda')$ worthwhile if $\lambda'$ is a worthwhile perturbation of $\lambda$, $t\in \lambda'(e)$ and $t\not\in\lambda(e)$.
\end{definition}
\medskip
\begin{lemma}\label{lem:min-worthwhile}
    Let $\lambda'$ be a perturbation of $(G,\lambda)$ such that the maximum reachability of $(G,\lambda')$ is $h$. Suppose there is no $\lambda''$ which perturbs fewer edges with maximum reachability of $(G,\lambda'')$ at least $h$. Then $\lambda'$ is a worthwhile perturbation of $\lambda$.
\end{lemma}
\begin{proof}
    We prove the lemma statement by contrapositive. Suppose that $\lambda'$ is not a worthwhile perturbation of $\lambda$ and the maximum reachability of $(G,\lambda')$ is $h$. Then, there must be a time-edge in $(G,\lambda')$ which we can revert to $\lambda$ without decreasing the maximum reachability of the temporal graph. Therefore, $\lambda'$ is not a perturbation which perturbs the minimum number of time-edges such that maximum reachability of the temporal graph is $h$.
\end{proof}
\begin{lemma}\label{lem:worthwhile-contained}
    Suppose that $\lambda'$ is a worthwhile perturbation of $(G,\lambda)$. Then, for every vertex $v$ with maximum reachability in $(G,\lambda')$, the subgraph induced by $\text{reach}((G,\lambda'),v)$ contains all edges perturbed by $\lambda'$.
\end{lemma}
\begin{proof}
    Suppose, for a contradiction, that $\lambda'$ is a worthwhile perturbation of $(G,\lambda)$ and there is a vertex $v$ in $(G,\lambda')$ such that $v$ has maximum reachability and the subgraph induced by $\text{reach}((G,\lambda'),v)$ does not contain all edges perturbed by $\lambda'$. By our assumption, there must be a perturbed edge $e$ whose endpoints are not both reachable from $v$. This is because, under $\lambda'$ there is no temporal path from $v$ which uses the edge $e$. Therefore, reverting this edge to its original temporal assignment under $\lambda$ cannot decrease the cardinality of the reachability set. This contradicts our definition of a worthwhile assignment.
\end{proof}
\begin{lemma}\label{lem:worthwhile-forest}
    The perturbed time-edges in a worthwhile perturbation $\lambda'$ of $(G,\lambda)$ form a forest.
\end{lemma}
\begin{proof}
    Suppose, for a contradiction, that there is a cycle $C$ of perturbed edges under a worthwhile perturbation $\lambda'$. Let $v$ be a vertex with maximum reachability in $(G,\lambda')$. Then, by Lemma~\ref{lem:worthwhile-contained}, the edges perturbed by $\lambda'$ are in the subgraph induced by $\text{reach}((G,\lambda'),v)$. Hence the vertices in $C$ are all temporally reachable from $v$. 
    
    We have two cases to consider. In the first, assume that for each perturbed time-edge, there is a temporal path from $v$ using this time-edge. Then there is a vertex $u$ in $C$ that $v$ can reach by two distinct temporal paths. This includes the trivial path if $v=u$. Reverting the final time-edge $(e_u,t)$ to $\lambda(e_u)$ of a path that arrives latest cannot decrease the cardinality of the reachability set of $v$. Since $v$ is a vertex with maximum reachability in $(G,\lambda')$, this means that reverting this time-edge does not decrease the maximum reachability of the temporal graph. Thus, $\lambda'$ cannot be worthwhile. Second, we consider the case where there is a perturbed time-edge such that there is no temporal path from $v$ including this time-edge. Reverting this time-edge trivially cannot change the reachability set of $(G,\lambda')$. Hence, we cannot have a cycle of perturbed time-edges in a worthwhile perturbation.
\end{proof}
We now show that, if the budget of perturbations is large enough, TRP and TRLP are equivalent problems.
\medskip
\begin{lemma}\label{lem:trp-trlp-big-zeta}
    If $\zeta\geq h-1$, there is a solution for \trlpp on $(G,\lambda)$ and integers $h$, $\delta$ if and only if there is a solution for \trppp on the same temporal graph with $h$ and $\delta$.
\end{lemma}
\begin{proof}
If $((G,\lambda),h,\delta,\zeta)$ is a yes-instance of \trlpp, then this trivially implies $((G,\lambda),h,\delta)$ is a yes-instance of \trppp. Let $\lambda'$ be a solution for this instance of \trppp which perturbs the fewest edges and $v$ be a vertex in $(G,\lambda')$ with reachability at least $h$. Let the number of perturbed edges under $\lambda'$ be $\zeta'$. By Lemma~\ref{lem:min-worthwhile}, $\lambda'$ must be a worthwhile perturbation. Furthermore, Lemmas~\ref{lem:worthwhile-contained} and~\ref{lem:worthwhile-forest} state that the perturbed time-edges must form a forest and be in the subgraph induced by $\text{reach}((G,\lambda'),v)$ where $v$ is a vertex with maximum reachability in $(G,\lambda')$. Therefore $v$ must reach both endpoints of every perturbed time-edge and the maximum reachability of $(G,\lambda')$ must be at least $\zeta'+1$.

At this point, if $\zeta\leq h-1$ we are done. Now, suppose otherwise. We will find a subset of these perturbations which achieves maximum reachability at least $h$. Let $\zeta'> h-1$, then $v$ must have reachability strictly greater than $h$. To find a perturbation of at most $h-1$ time-edges such that $v$ has temporal reachability at least $h$, we let $u_1$ and $u_2$ be the two endpoints of a perturbed time-edge and calculate $\gamma=\max\{t_{Fo}((G,\lambda'),v,u_1),t_{Fo}((G,\lambda')v,u_2)\}$ and order the time-edges by these values. We then keep the $h-1$ time-edges with smallest $\gamma$ values and revert all other time-edges. Since we only revert time-edges whose endpoints are first reached after the time-edges whose perturbations we keep, both endpoints of the remaining perturbed time-edges must still be reachable from $v$. Recall that these time-edges cannot form any cycles. Therefore, since $v$ reaches both endpoints of every perturbed edge, $v$ must reach at least $\zeta+1\geq h$ vertices. Thus $((G,\lambda),\delta,h)$ is a yes-instance of \trppp if and only if $((G,\lambda),\delta,h,\zeta)$ is a yes-instance of \trlpp for $\zeta\geq h-1$.
\end{proof}
Combining Theorem~\ref{thm:temp-reach-perturbed} and Lemma~\ref{lem:trp-trlp-big-zeta} gives us the following.
\medskip
\begin{theorem}\label{thm:trlp-big-zeta}
    We can solve \trlpp in $O(nm (\log n + \log \delta \tau (G,\lambda)))$ time when $\zeta\geq h-1$.
\end{theorem}

\section{Structural Restrictions on the Underlying Graph}\label{sec:structural}
We now consider TRLP when the underlying graph $G$ has restricted structure: in particular when $G$ is a tree, or has bounded treewidth. 
When the input underlying graph is a tree, we solve TRLP by a dynamic program which uses Multiple Choice Knapsack to choose the perturbations which maximise the reachability from the vertex we are considering in time  $O(n^2\zeta^3(T(G,\lambda)+\delta))$.  When $G$ has constant bounded treewidth and the lifetime and size of perturbations allowed are also bounded, we use a dynamic program over a tree decomposition to solve TRLP in polynomial time.

\subsection{Knapsack-based dynamic program for trees}
The formal definition of the knapsack problem that we use is as follows.

\begin{framed}
\noindent
\textbf{\textsc{Multiple Choice Knapsack Problem (MCKP)}}\\
\textit{Input:} An integer $c$ and $q$ mutually disjoint classes $N_1,\ldots,N_q$ of elements such that each element $j\in N_i$ has a profit $p_{ij}$ and weight $w_{ij}$.\\
\textit{Output:} A function $\phi:\bigcup_iN_i\to \{0,1\}$ such that $\sum_{j\in N_i}\phi(j)=1$ for all $N_i$, $\sum_{i=1}^{q}\sum_{j\in N_i}w_{ij}\phi(j)\leq c$, and $\sum_{i=1}^{q}\sum_{j\in N_i}p_{ij}\phi(j)$ is maximised.
\end{framed}

Dudzi\'nski and Walukiewicz~\cite{dudzinski_exact_1987} show that this can be solved by a dynamic program in time $O(c\sum^q_{i=1}n_i)=O(nc)$ where $n_i$ is the size of the class $N_i$ and $n$ is the total number of elements. In our case, the size of each class is bounded by $\zeta$.

We solve \trlpp by a bottom-up dynamic programming argument. The algorithm checks whether there is a $(\delta,\zeta)$-perturbation $\lambda'$ such that a given source vertex $v_s$ has reachability at least $h$. We root the tree at the source, giving us an orientation so that we can refer to the parents and children of vertices. Since we are asking if there exists a $(\delta,\zeta)$-perturbation such that $R_{\max}(G,\lambda)\geq h$, we can repeatedly run the algorithm for different choices of source vertex. Once we reach the root, if there exists a state which corresponds to reachability at least $h$, we accept. Otherwise, we reject and try another source vertex.
    
For any vertex $v$ in the rooted tree $G$, we define a state as the triple $(\zeta_v,r_v,t_v)$ where $0\leq \zeta_v\leq \zeta$, $0<r_v\leq h$ and $0\leq t\leq T(G,\lambda)+\delta$. We say that a state $(\zeta_v,r_v,t_v)$ is \emph{valid} if and only if there exists a $(\delta,\zeta_v)$-perturbation of the temporal subtree $G_v$ rooted at $v$ where $v$ reaches $r_v$ vertices after time $t_v$. By ``after time $t_v$'' we mean that the vertices are reachable by a temporal path which departs $v$ at time $t_v$ at the earliest. We say that the state is \emph{supported} by this perturbation.  We further say that a state $(\zeta_v,r_v,t_v)$ is \emph{maximally valid} if, for every other valid state $(\zeta_v,r',t_v)$ of $v$, we have $r' < r_v$.

    
    
Given all maximally valid states of the children of a vertex $v$, for each time $t$ we can calculate all maximally valid states $(\zeta_v,r_v,t)$ of $v$ in $O(d_c\zeta^2)$ time where $d_c$ is the number of children of $v$. To do this we use the algorithm of Dudzi\'nski and Walukiewicz for MCKP \cite{dudzinski_exact_1987}.

The algorithm starts by assigning every leaf the set of states $\{(0,1,t)\colon t \in [0,T(G,\lambda)+\delta]\}$. For any non-leaf vertex $v$ and every time $t\in [0,T(G,\lambda)+\delta]$, we check if the state $(\zeta_v,r_v,t)$ is valid and then easily remove states that are not maximally valid. The intuition behind our method is that we use MCKP to find an allocation of perturbations to the subtrees rooted at each child of $v$ which maximises the reachability of $v$.
    \begin{enumerate}
        \item For each child $c$, we make a list of pairs $(\zeta_c,r_c)$. A pair $(\zeta_c,r_c)$ is present if
        \begin{enumerate}
            \item $\zeta_c=0$ and $r_c=1$; or
            \item there is a time $t'$ after $t$ at which the edge $vc$ is active and there exists a valid state $(\zeta_c,r_c-1,t_c)$ of $c$ such that $t_c>t'$; or
            \item there is a valid state $(\zeta_c-1,r_c-1,t')$ of $c$ where $t'$ is a time after $t$ and $vc$ is active in the interval $[t'-\delta,t'+\delta]$.
        \end{enumerate}
        \item If multiple pairs associated with a single child $c$ have the same $\zeta_c$ value, we remove all but the one with the highest $r_c$ value.
        \item  For each integer $w$ with $0\leq w\leq \zeta$, we can find a combination $S$ of pairs in $O(n\zeta^2)$ time where: one pair is chosen per child; the sum of $\zeta_c$ is at most $w$; and the sum of $r_c$ is maximised. Each set of the input of MCKP consists of the pairs belonging to a given child vertex. For a pair $(\zeta,r)$, the profit is $r$ and the weight is $\zeta$. More formally, our input to an algorithm which solves MCKP is the integer $c=w$ and the set of classes $N_1,\ldots,N_{d_c}$ where a class $N_i$ corresponds to the $i$th child. For each pair $(\zeta_c,r_c)$ of a child $c$, we add an element with weight $\zeta_c$ and profit $r_c$ to the corresponding class. We can then solve MCKP for each $w$ such that $1\leq w\leq\zeta$. The state $(w,\max(r),t)$ where $\max(r)$ is the maximum value of profit as found by an algorithm solving MCKP with our inputs is then a valid state of $v$ and records the maximum reachability in the subtree rooted at $v$ given $w$ perturbations.
        \item Delete any states with $\zeta_v>\zeta$ or $t_c<0$ and any states $(\zeta_v,r_v',t_v)$ where there exists a state $(\zeta_v,r_v,t_v)$ such that $r_v>r_v'$.
    \end{enumerate}
\medskip
\begin{lemma}\label{lem:tree-trlp-correctness}
    For every vertex $v$, a state $(\zeta_v,r_v,t_v)$ is found by the algorithm above if and only if it is maximally valid for $v$.
\end{lemma}
\begin{proof}
    We prove this statement by induction on the height of $v$. The height of a vertex $v$ is defined as the maximum distance to a leaf in the subtree rooted at $v$. 
    For the base case, it is clear that the set of states $\{(0,1,t)\colon t \in [0,T(G,\lambda)+\delta]\}$ contains precisely the maximally valid states for a leaf.

    We now assume that the lemma statement is true for all descendants of $v$. In particular, we assume that for every child $c$ of $v$, the algorithm finds a state $(\zeta_c,r_c,t_c)$ if and only if that state is maximally valid for $c$.  
    
    Suppose first that the algorithm returns the state $(\zeta_v,r_v,t_v)$ for $v$; we will show that this state is maximally valid for $v$.  By our computation of states, there must be a set $S$ of valid pairs $(\zeta_c,r_c)$ of the children consisting of one state per child such that the sum of $r_c$ values is maximised and the sum of $\zeta_c$ values is at most $\zeta_v$. These child pairs either correspond to a maximally valid child state with a $t$ value later than $t_v$ or account for a perturbation on the edge between $v$ and the child.

    By the inductive hypothesis, each of the child states $(\zeta_c,r_c,t_c)$ used to construct a state of $v$ must be supported by a $(\delta,\zeta_c)$-perturbation. Let $\lambda'$ be a perturbation of $\lambda$ which, when restricted to each subtree, supports the state $(\zeta_c,r_c,t_c)$. We also ensure that, if $t_c$ is earlier than $vc$ is active, the edge $vc$ is perturbed under $\lambda'$ so that it is active before $t_c$. This must be achievable by a $\delta$-perturbation by our construction. Then, under $\lambda'$, $v$ must reach each of the vertices reachable from its children after time $t_c$ in the subtree rooted at $v$. This gives us a $(\delta,\zeta_v)$-perturbation such that $v$ reaches $r_v$ vertices after $t_v$.
    
    Now assume the state $(\zeta_v,r_v,t_v)$ is maximally valid for $v$; we will show that the algorithm finds this state.  Let $\lambda^*$ be a $(\delta,\zeta_v)$-perturbation of the subtree induced at $v$ which supports $(\zeta_v,r_v,t_v)$.  Denote by $\zeta_c$ the number of permutations made by $\lambda^*$ in the subtree rooted at a child $c$ of $v$. Let $t_c$ be the earliest time after $t_v$ at which the edge $vc$ is active. Note that the number of vertices reached from $c$ after time $t_c$ under these perturbations must be maximal over the perturbations of $(G_c,\lambda|_{E(G_c)})$, otherwise $\lambda^*$ would not maximise the vertices reached.  By the inductive hypothesis, there must be maximally valid states $(\zeta_c,r_c,t_c)$ for each child $c$ of $v$. These states correspond to pairs which maximise the profit in our instance of MCKP. Therefore, the state $(\zeta_v,r_v,t_v)$ is found by the algorithm in Step 3; moreover, by maximality, it will not be deleted in Step 4.

    Hence the algorithm finds a state $(\zeta_v,r_v,t_v)$ if and only if it is maximally valid.
    \end{proof}
    This proves correctness of the algorithm.
    \medskip
\begin{theorem}\label{thm:tree-alg-trlp}
    \trlpp is solvable in $O(n^2\zeta^3(T(G,\lambda)+\delta))$ time on temporal graphs where the underlying graph $G$ is a tree.
\end{theorem}
\begin{proof}
    Lemma~\ref{lem:tree-trlp-correctness} shows correctness of our algorithm. This algorithm runs in $O(n^2\zeta^3(T(G,\lambda)+\delta))$ time. To see this, note that for any vertex there are at most $(\zeta+1)(T(G,\lambda)+\delta)$ maximally valid states. That is, there are at most $\zeta+1$ states per time $t\in T(G,\lambda)+\delta$. For the non-leaf vertices, we solve MCKP to find the maximally valid states. This is done in $O(n\zeta^2)$ time. All other computation done for each state of a given vertex is achievable in constant time. Thus, the algorithm takes at most $O(n\zeta^3(T(G,\lambda)+\delta))$ time per vertex. If we perform this on each vertex, then we solve TRLP in $O(n^2\zeta^3(T(G,\lambda)+\delta))$ time.
\end{proof}

\subsection{Tree Decomposition Algorithm}

In this section, we prove the following result.
\medskip
\begin{theorem}\label{thm:trlp-tree-decomp}
    Given a temporal graph $(G,\lambda)$ where the treewidth of $G$ is at most $\omega$, we can solve TRLP in time
    $$O \left( n\omega^4 \zeta^3 (T(G,\lambda)+\delta)^{\omega})  \cdot (2 \delta h  T(G,\lambda) )^{4\omega^2 (T(G,\lambda) + \delta)^\omega}   \right).$$
    In particular, we can solve TRLP in polynomial time when the treewidth of the underlying graph, the permitted perturbation size $\delta$ and the lifetime of $(G,\lambda)$ are all bounded by constants.
\end{theorem}
Our algorithm proceeds by dynamic programming over a nice tree decomposition.  Before giving the algorithm, we discuss some definitions and the construction of our states.

\subsubsection{Algorithm preliminaries}

We begin by defining tree decompositions. We note that we use \emph{node} when referring to the tree decomposition and \emph{vertex} when referring to the original graph for clarity.
\medskip
\begin{definition}[Cygan et al. \cite{Cygan2015}]
    A \emph{tree decomposition} of a static graph $G$ is a pair $\mathcal{T}=(T,\{\mathcal{D}(s)_{s\in V(T)}\})$ where $T$ is a tree where every node is assigned a \emph{bag} $\mathcal{D}(s)\subseteq V(G)$ such that the following conditions hold:
    \begin{enumerate}[leftmargin=0.5cm]
        \item $\bigcup_{s\in V(T)}\mathcal{D}(s)=V(G)$;
        \item for every $uv\in E(G)$, there exists a node $s$ of $T$ such that a bag $\mathcal{D}(s)$ contains both $u$ and $v$;
        \item for every $u\in V(G)$, the set $\{s\in V(T)\,:\, u\in \mathcal{D}(s)\}$ induces a connected subtree of $T$.
    \end{enumerate}
    The \emph{width} of a tree decomposition is $\max_{s\in V(T)}|\mathcal{D}(s)|-1$. The \emph{treewidth} of a graph is the minimum width of a tree decomposition of that graph.
\end{definition}
\medskip
We use $G_s$ to refer to the subgraph induced by vertices introduced in the subtree rooted at $s$. Let $G'_s$ denote the subgraph induced by $V(G_s)\setminus \mathcal{D}(s)$.

For our algorithm we use a variation of a tree decomposition known as a nice tree decomposition. This allows us to reduce the number of cases we need to consider when computing the dynamic program over the decomposition.
\medskip
\begin{definition}[Cygan et al. \cite{Cygan2015}]\label{def:nice-tree-decomp}
    A rooted tree decomposition $(T,\{\mathcal{D}(s)_{s\in V(T)}\})$ is \emph{nice} if the following are true
    \begin{itemize}[leftmargin=0.5cm]
        \item $\mathcal{D}(\ell)=\emptyset$ for all leaf nodes $\ell$ of $T$;
        \item every non-leaf node of $T$ is one of the following three types:
        \begin{itemize}
            \item (\textbf{Introduce}) a node $s$ with exactly one child $s'$ such that $\mathcal{D}(s)=\mathcal{D}(s')\cup \{v\}$ for some vertex $v\notin\mathcal{D}(s)$;
            \item (\textbf{Forget}) a node $s$ with exactly one child $s'$ such that $\mathcal{D}(s)=\mathcal{D}(s')\setminus \{v\}$ for some vertex $v\in\mathcal{D}(s')$;
            \item (\textbf{Join}) a node $s$ with exactly two children $s_1,s_2$ such that $\mathcal{D}(s)=\mathcal{D}(s_1)=\mathcal{D}(s_2)$.
        \end{itemize}
    \end{itemize}
\end{definition}

It is well known (see, for example, Cygan et al. \cite[Lemma 7.4]{Cygan2015}) that if a graph $G$ admits a tree decomposition of width at most $w$, then it must admit a nice tree decomposition of width at most $w$ with at most $O(wn)$ nodes. Furthermore, they state that a nice tree decomposition can be computed in $O(\omega^2\cdot\max(n,|V(T)|))$ time, given a tree decomposition of the graph. We solve TRLP by dynamic program on a nice decomposition graph.

Without loss of generality, we assume that the root of the tree decomposition is a bag containing only a source vertex $x$. The algorithm finds the maximum reachability of the source vertex under any $(\delta,\zeta)$-perturbation. To find the maximum reachability of the graph under any $(\delta,\zeta)$-perturbation, we can run the algorithm with each possible source vertex. The dynamic program operates by computing a set of states for each bag. Calculation of states requires us to restrict departure and arrival times of temporal paths. Recall that the \emph{arrival time} of a temporal path is the time of the last time-edge in the path. Similarly, the time at which a temporal path \emph{departs} is the time $t$ if first time-edge in the path occurs at $t'\geq t$. For vertices $v$ and $u$ in the bag, if there is no temporal path from a vertex $v$ to a vertex $u$, we use the convention that the foremost arrival time of a path from $v$ to $u$ is $\infty$. Our states record:
\begin{itemize}
    \item the times at which edges in the bag are active following perturbation;
    \item for each vertex $v$ in the bag and time $t$, the foremost arrival time of a path in $G_s$ to each vertex in the bag from $v$ which departs at time $t$;
    \item for each set of vertices in the bag where each vertex $v$ has a specified (possibly different) departure time $t$, the total number of distinct vertices reached below the bag by vertices by a path departing each $v$ at time $t$; and
    \item the number of edges perturbed below the bag.
\end{itemize}

To describe the times at which paths leave each vertex in a given set of vertices, we use the following notation. Let $\mathcal{U}_{\mathcal{D}(s)}$ be the set of functions which map the vertices in the bag $\mathcal{D}(s)$ to $[0,T(G,\lambda)+\delta]\cup \{\infty\}$. For notational convenience, we consider a function $U$ in $\mathcal{U}_{\mathcal{D}(s)}$ to be the set $\{(v,U(v)): v \in S\subseteq\mathcal{D}(s)\}$ of vertex-time pairs. Let $(G,\lambda)$ be a temporal graph and $S\subseteq V(G)$, we use $\mathcal{E}((S,\lambda))$ to denote the set of time-edges in the subgraph induced by the vertices in the set $S$. The restriction of a function $f$ to some domain $H$ is denoted $f|_H$.

A state is a tuple $(P_s,\mathcal{R}^{\text{below}},\mathcal{R}^{\text{in}},\zeta_{G'_s})$ where 
\begin{enumerate}
    \item $P_s$ is a function $P_s:E(G[\mathcal{D}(s)])\to \mathcal{P}([1,T(G,\lambda)+\delta])$ such that $|\lambda(e)|=|P_s(e)|$ for all $e \in E(G)$, where $\mathcal{P}$ denotes the power set; this describes the  time-edges in the bag following a perturbation;
    \item $\mathcal{R}^{\text{in}}:\mathcal{D}(s)\times[0,T(G,\lambda)+\delta]\to \mathcal{U}_{\mathcal{D}(s)}$ is a function which maps a vertex-time pair $(v,t)$ to a function which gives the foremost arrival time of a path in $G_s$ from $v$ departing at $t$ to each vertex in the bag;
    \item for each $U\in \mathcal{U}_{\mathcal{D}(s)}$, $\mathcal{R}^{\text{below}}:\mathcal{U}_{\mathcal{D}(s)}\to [0,h]$ gives the number of distinct vertices below the bag to which there exists a temporal path from some vertex $v$ departing at time $U(v)$; and
    \item $\zeta_{G'_s}\in [0,\zeta]$ gives the number of perturbations allowed on edges with at least one endpoint forgotten below this bag.
\end{enumerate}
Since we only want to know whether the number of vertices reached is at least $h$, the set-function pair is mapped to $h$ by $R^{\text{below}}$ if there are at least $h$ vertices reachable from vertices in the bag when the earliest permitted departure time for a vertex $v$ is $U(v)$. Note that, for $\mathcal{R}^{\text{in}}(u,t)$, we allow foremost arrival at vertices by paths traversing vertices in the bag and vertices which have been forgotten. That is, $\mathcal{R}^{\text{in}}(u,t)$ describes the set of vertices in $\mathcal{D}(s)$ reachable from $u$ after time $t$ by paths consisting of vertices in $G_s$. 

We say that a state $(P_s,\mathcal{R}^{\text{below}},\mathcal{R}^{\text{in}},\zeta_{G'_s})$ is valid if there exists a perturbation $\lambda'$ of $(G_s,\lambda|_{G_s})$ such that
\begin{enumerate}
    \item under $\lambda'$, $\mathcal{R}^{\text{in}}(v,t)$ gives the foremost arrival times in $(G_s,\lambda')$ at each vertex in $\mathcal{D}(s)$ from $v$ departing at $t$;
    \item $\lambda'$ is equal to $P_s$ when restricted to $\mathcal{D}(s)$;
    \item $\lambda'$ perturbs exactly $\zeta_{G'_s}$ edges in $G'_s$;
    \item for each assignment of times $U\in\mathcal{U}_{\mathcal{D}(s)}$, the cardinality of the union of vertices in $G_s'$ reached from all vertices $v\in \mathcal{D}(s)$ by paths departing at time $U(v)$ is $\min(h,\mathcal{R}^{\text{below}}(U))$;
    \item the total number of edges perturbed in $G_s$ is at most $\zeta$.
\end{enumerate}
If this is the case we say that $\lambda'$ \emph{supports} $(P_s,\mathcal{R}^{\text{below}},\mathcal{R}^{\text{in}},\zeta_{G'_s})$. For ease, we refer to the set of time-edges perturbed by $P_s$ as $\text{changed}(P_s)$. That is, $\text{changed}(P_s)=\{(e,t)\,:\,t\in P_s(e), t\notin\lambda(e)\}$. Note that the final condition requires that $\text{changed}(P_s)+ \zeta_{G'_s}\leq \zeta$.

%

We now describe how we determine the validity of states for each type of node.

\subsubsection*{Leaf nodes} 
We assume that a leaf node contains no vertices by the definition of a nice tree decomposition (Definition~\ref{def:nice-tree-decomp}). In addition, there are no edges below a leaf which can be perturbed. Therefore, there is only one trivial possibility for a valid leaf state.

\subsubsection*{Introduce nodes}
Suppose that an introduce node $s$ has child $s'$ and that $\mathcal{D}(s)\setminus\mathcal{D}(s')=\{u\}$. We determine whether a state of an introduce node is valid as follows. 
\medskip
\begin{lemma}\label{lem:reach-nice-tree-intro}
    Let $s$ be an introduce node with child $s'$ and introduced vertex $u$. A state $(P_s,\mathcal{R}^{\text{below}}_s, \mathcal{R}^{\text{in}}_s,\zeta_{G'_s})$ of an introduce node $s$ is valid if and only if there is a valid state $(P_{s'},\mathcal{R}^{\text{below}}_{s'}, \mathcal{R}^{\text{in}}_{s'},\zeta_{G'_{s'}})$ of its child $s'$ such that
\begin{enumerate}[leftmargin=0.5cm]
    \item $P_{s}|_{E(G[\mathcal{D}(s')])}=P_{s'}$;
    \item for all vertices $v$ in $\mathcal{D}(s)$ and times $t$, $\mathcal{R}^{\text{in}}_s(v,t)(u)=t''$ if and only if either $v=u$ and $t=t''$, or $t''$ is the earliest time such that there exist $t'<t''$ and $w\in N(u)\cap \mathcal{D}(s)$ where $\mathcal{R}^{\text{in}}_{s'}(v,t)(w)=t'$ and $t''\in P_s(wu)$;
    \item for all $v\in \mathcal{D}(s)$ such that $v\neq u$, each time $t$, and each neighbour $w$ of $u$ in $\mathcal{D}(s)$, if $t_w$ is defined to be $\min\{t'\in P_s(uw)\,:\, t'\geq t\}$, then 
      $$ \mathcal{R}^{\text{in}}_{s}(u,t)(v) = \left\{
     \begin{array}{lr}
     \min\left(\min_{w\in N(u)\setminus \{v\}}\mathcal{R}^{\text{in}}_{s'}(w,t_w+1)(v), t_v\right) & \text{if $v\in N(u)$}\\
     \min_{w\in N(u)\setminus \{v\}}\mathcal{R}^{\text{in}}_{s'}(w,t_w+1)(v) & \text{otherwise}
     \end{array}
   \right.;$$
    \item for each $v,x\neq u\in\mathcal{D}(s)$ and each time $t$, if we let $t'$ be the earliest time such that there exists a neighbour $w$ of $u$ where $t'\in P_s(wu)$ and $\mathcal{R}^{\text{in}}_{s'}(v,t)(w)<t'$ or $\infty$ if no such neighbour exists, then $\mathcal{R}^{\text{in}}_{s}(v,t)(x)=\min(\mathcal{R}^{\text{in}}_{s'}(v,t)(x), \mathcal{R}^{\text{in}}_{s}(u,t'+1)(x))$;
    \item for all $U$ in $\mathcal{U}_{\mathcal{D}(s)}$, $\mathcal{R}^{\text{below}}_s(U)=\mathcal{R}^{\text{below}}_{s'}(\min(U,\mathcal{R}^{\text{in}}_s(u,U(u))))$;
    \item $\zeta_{G'_s}= \zeta_{G'_{s'}}$;
    \item $|\text{changed}(P_s)|+\zeta_{G'_{s}}\leq \zeta$.
\end{enumerate}

\end{lemma}
\begin{proof}
    We claim that a state $(P_s,\mathcal{R}^{\text{below}}_s, \mathcal{R}^{\text{in}}_s ,\zeta_{G'_s})$ is supported by a $(\delta,\zeta)$-perturbation of $(G,\lambda)$ if and only if the restriction of that perturbation to $G_{s'}$ supports a state $(P_{s'},\mathcal{R}^{\text{below}}_{s'}, \mathcal{R}^{\text{in}}_{s'} ,\zeta_{G'_{s'}})$ of $s'$ such that the criteria in the lemma hold. We begin by showing that the state $(P_{s},\mathcal{R}^{\text{below}}_{s}, \mathcal{R}^{\text{in}}_{s} ,\zeta_{G'_{s}})$ is supported by a $(\delta,\zeta)$-perturbation $\lambda'$ if the criteria hold. Let $\lambda''$ be the perturbation which supports the state $(P_{s'},\mathcal{R}^{\text{below}}_{s'}, \mathcal{R}^{\text{in}}_{s'} ,\zeta_{G'_{s'}})$ and $\lambda'$ be the extension of $\lambda''$ given by $P_s$.

    If criteria (1), (6), and (7) hold, then $\lambda'$ must be an extension of $\lambda''$ and it must be a $(\delta,\zeta)$-perturbation of $\lambda|_{G_s}$. We now show that if the remaining criteria hold ((2)-(5)), then $\lambda'$ must support the state $(P_s,\mathcal{R}^{\text{below}}_s, \mathcal{R}^{\text{in}}_s ,\zeta_{G'_s})$. We begin by showing that the foremost arrival time of a path from $v\in \mathcal{D}(s)$ to $y\in \mathcal{D}(s)$ departing at time $t$ under $\lambda'$ is $t'$ if and only if $\mathcal{R}^{in}_s(v,t)(y)=t'$ as given in criteria (2)-(4). We have 3 cases to consider, namely $u=y$, $u=v$, and $v\neq u\neq y$. 
    
    It is clear that if $u=v=y$, the foremost arrival time of a path departing at $t$ under any temporal assignment from $v$ to itself is $t$. To check the remainder of the first case, we see that if $v\neq u= y$, the foremost arrival time of a path from $v$ to $u$ departing at time $t$ under $\lambda'$ must be the earliest time $t''$ such that there exists a neighbour $w$ of $u$ with foremost arrival time $t'<t''$ of a path departing at $t$, and an appearance of $\lambda(uw)$ at $t''$. That is, we require a path to $w$ to which we can prefix the edge $wu$. Since $\lambda'$ coincides with $P_s$ and supports $\mathcal{R}^{\text{in}}_{s'}$, this is precisely criterion (2). 
    
    Furthermore, suppose that for all vertices $v$ in $\mathcal{D}(s)$ and times $t$, $\mathcal{R}^{\text{in}}_s(v,t)(u)=t'$ if and only if either $v=u$ and $t=t''$, or $t''$ is the earliest time such that there exist $t'<t''$ and $w\in N(u)\cap \mathcal{D}(s)$ where $\mathcal{R}^{\text{in}}_{s'}(v,t)(w)=t'$ and $t''\in P_s(wu)$. Then we either have the case where the foremost path to $u$ in $(G_s,\lambda')$ is trivial or the concatenation of a foremost path to a neighbour of $u$ and an edge from that neighbour to $u$. Thus, the foremost path from a vertex $v\in\mathcal{D}(s)$ departing at time $t$ to $u$ arrives at time $t''$ if and only if $\mathcal{R}^{\text{in}}_s(v,t)(u)=t''$.

    In the second case, we have that $u=v\neq y$. If there exists a foremost path from $u$ to $y\in \mathcal{D}(s)$ under $\lambda'$ which departs at some time $t_1$ and arrives at time $t_2$, then, since $u$ has no neighbours in $G_{s}\setminus\mathcal{D}(s)$, there must be a (possibly trivial) path from a neighbour $w$ of $u$ to $y$ in $\mathcal{G}_{s'}$ which departs at some time $t_1'>t_1$ and arrives at time $t_2$. This must be the foremost path from a neighbour of $u$ to $y$ departing at $t_1$ such that the edge $uw$ is active at a time $t'$ such that $t_1\leq t'<t'_1$. Therefore, criterion (3) describes the foremost arrival time at vertices in the bag of paths from $u$ departing at time $t$ under $\lambda'$.

    Suppose that $\mathcal{R}^{\text{in}}_s(u,t)(y)=t''$. Then, if $y$ is a neighbour of $u$, $t''$ is the minimum of the first appearance of the edge $uy$ at or after $t$ and the earliest arrival time of the foremost path from another neighbour $w$ of $u$ to $y$ which does not traverse $u$ and departs after the earliest appearance of $uw$ at or after $t$. If $y$ is not a neighbour of $u$, then earliest arrival time of the foremost path from any neighbour $w$ of $u$ to $y$ which does not traverse $u$ and departs after the earliest appearance of $uw$ at or after $t$ must be the foremost arrival time of a path from $u$ to $y$.
    Hence, the foremost path from a vertex $u\in\mathcal{D}(s)$ departing at time $t$ to $y$ arrives at time $t''$ if and only if $\mathcal{R}^{\text{in}}_s(u,t)(y)=t''$.
    
    For the final case, let $v\neq u\neq y$ and $t''$ be the earliest time of arrival of a path from $v$ to $y$ departing at time $t$ under $\lambda'$. This path either traverses $u$ or it does not. If it does, then the foremost arrival time of the path can be found by finding the foremost arrival time of a path from $u$ to $y$ departing after the arrival of the foremost path from $v$ to $u$. If the path does not traverse $u$, $t''$ must be the same as the foremost arrival time in the graph $(G_{s'},\lambda'')$. Since we take the minimum of both values in (4), $t''$ must be equal to $\mathcal{R}^{\text{in}}_s(v,t)(y)$ for all $v$, $y\in\mathcal{D}(s)$ and times $t$. Therefore, if criteria (2)-(4) hold, there must be a $(\delta,\zeta)$-perturbation $\lambda'$ consistent with $P_s$ where the foremost arrival time under $\lambda'$ at each vertex in $\mathcal{D}(s)$ by paths from vertices $v$ in the bag departing at time $t$ is given by $\mathcal{R}^{\text{in}}_s(v,t)$ for all $v\in\mathcal{D}(s)$ and times $t$.

    We now show that, for each function $U\in\mathcal{U}_{\mathcal{D}(s)}$ the number of distinct vertices reached from vertices in the bag by paths departing $v\in\mathcal{D}(s)$ at time assigned by $U(v)$ under $\lambda'$ is $\mathcal{R}^{\text{below}}_{s'}(\min(U,\mathcal{R}^{\text{in}}_s(u,U(u))))$. This implies that if all criteria hold, the state $(P_s,\mathcal{R}^{\text{below}}_s, \mathcal{R}^{\text{in}}_s ,\zeta_{G'_s})$ is supported by $\lambda'$ and thus valid.

    To do this, we argue that for all functions $U\in \mathcal{U}_{\mathcal{D}(s)}$, the number of vertices reached in $G_s$ under $\lambda'$ by the set $A$ of vertices assigned finite times by $U$ by paths departing at the times assigned by $U$ is exactly $\mathcal{R}^{\text{below}}_s(U)$ if $\mathcal{R}^{\text{below}}_s(U)<h$ and at least $h$ otherwise. Suppose, for contradiction, that under $\lambda'$ the set $A$ reaches $h'$ vertices in $G'_s$ at the times given by $U$ and either $h>h'\neq \mathcal{R}^{\text{below}}_s(U)$, or $\mathcal{R}^{\text{below}}_s(U) < h\leq h'$. The vertices reachable below the bag must be reachable by a path terminating at a vertex in $\mathcal{D}(s')$ which prefixes a path which only traverses vertices in $G'_s$. The vertex $u$ cannot be incident to any forgotten vertices by the construction of a nice tree decomposition, therefore the final vertex $v$ in $\mathcal{D}(s)$ on this path must be in $\mathcal{D}(s')$. The path from $v$ must depart at either $U(v)$ or the time of arrival of the path from $u$ to $v$ (or both). In addition, no vertices are forgotten between $s$ and $s'$. Hence, $G'_s=G'_{s'}$. This gives us that, under $\lambda'$, the number of vertices reached strictly below the bag by a set $A$ at the times specified by the function $U$ must be $\mathcal{R}^{\text{below}}_s(U)=\mathcal{R}^{\text{below}}_{s'}(\min(U,\mathcal{R}^{\text{in}}_s(u,U(u))))$. If the number of vertices reachable from $A$ at the times given by $U$ not equal to $\mathcal{R}^{\text{below}}_s(U)$ and $\mathcal{R}^{\text{below}}_s(U)<h$, then the number of vertices in $G'_{s'}$ reachable by the set of vertices assigned finite times by the function $\min(U,\mathcal{R}^{\text{in}}(u,U(u))$ must not be equal to
    $\mathcal{R}^{\text{below}}_{s'}(\min(U,\mathcal{R}^{\text{in}}(u,U(u)))$. This contradicts validity of the state $(P_{s'},\mathcal{R}^{\text{below}}_{s'}, \mathcal{R}^{\text{in}}_{s'},\zeta_{G'_{s'}})$ of $s'$. This shows that if (5) holds, then $\mathcal{R}^{\text{below}}_s(U)$ describes, under $\lambda'$, the number of distinct vertices reached below the bag by vertices in the bag departing at the times specified by $U$. Hence, if all criteria in the lemma statement hold, the state $(P_{s},\mathcal{R}^{\text{below}}_{s}, \mathcal{R}^{\text{in}}_{s},\zeta_{G'_{s}})$ must be supported by $\lambda'$.

    We now show that if there exists a $(\delta,\zeta)$-perturbation $\lambda^*$ of $(G_s,\lambda|_{G_s})$ that supports the state $(P_{s},\mathcal{R}^{\text{below}}_{s}, \mathcal{R}^{\text{in}}_{s},\zeta_{G'_{s}})$, then the criteria must hold. Let $\lambda^*_{s'}$ be the restriction of $\lambda^*$ to $G_{s'}$, and $(P_{s'},\mathcal{R}^{\text{below}}_{s'}, \mathcal{R}^{\text{in}}_{s'} ,\zeta_{G'_{s'}})$ be the child state supported by $\lambda^*$. Then, if $P_s$ is the extension of $P_{s'}$ which coincides with $\lambda^*$, (1) must be true. In addition, by the definitions of a nice tree decomposition and a $(\delta,\zeta)$-perturbation, (6) and (7) must also hold.

    What remains is to show that (2)-(5) also hold. It is clear that (2) describes the foremost arrival time of paths to $u$ under $\lambda^*$. That is, the foremost arrival time from some vertex $v\in\mathcal{D}(s)$ to $u$ departing at $t$ is $t''$ if and only if $v=u$ and $t=t''$, or $t''$ is the earliest time such that there is a foremost path from $v$ to a neighbour $w$ of $u$ arriving at time $t'<t''$ and an appearance of the edge $uw$ at time $t''$. Thus this criterion must match the state $(P_{s},\mathcal{R}^{\text{below}}_{s}, \mathcal{R}^{\text{in}}_{s},\zeta_{G'_{s}})$. 
    
    The foremost arrival of paths from $u$ departing at $t$ to a vertex $v\neq u$ under $\lambda^*$ must be the minimum arrival time of a foremost path to $v$ from a neighbour $w$ of $u$ departing after the earliest occurrence of the edge $uw$ at or after $t$. Therefore, criterion (3) also matches the state $(P_{s},\mathcal{R}^{\text{below}}_{s}, \mathcal{R}^{\text{in}}_{s},\zeta_{G'_{s}})$. To see that criterion (4) also holds, let $t_2$ be the earliest time such that $\mathcal{R}^{\text{in}}_{s'}(v,t_1)(y)=t_2$ or $\mathcal{R}^{\text{in}}_s(u,t^*(u))(y)=t_2$ where $\mathcal{R}^{\text{in}}_s(v,t_1)(u)<t^*(u)$. Suppose that $\mathcal{R}^{\text{in}}_{s'}(v,t_1)(y)=t_2$. Then, since the restriction $\lambda^*_{s'}$ of $\lambda^*$ supports $(P_{s'},\mathcal{R}^{\text{below}}_{s'}, \mathcal{R}^{\text{in}}_{s'},\zeta_{G'_{s'}})$, there must be a foremost temporal path from $v$ to $y$ in $\mathcal{D}(s)$ under $\lambda^*$ which departs at $t_1$ and arrives by $t_2$. If $\mathcal{R}^{\text{in}}_s(u,t^*(u))(y)=t_2$ where $\mathcal{R}^{\text{in}}_s(v,t_1)(u)<t^*(u)$, then we have a path from $v$ to $u$ in $\mathcal{D}(s)$ which departs at $t_1$ and arrives before $t^*$ and a path from $u$ to $y$ which departs at $t^*$ and arrives by $t_2$. Concatenating these paths gives us a path under $\lambda^*$ which departs $v$ at $t_1$ and arrives at $y$ by $t_2$. We can see that $t_2$ must be the foremost arrival time of a path from $v$ to $w$ departing at $t$ under $\lambda^*$, else we contradict that $(P_{s'},\mathcal{R}^{\text{below}}_{s'}, \mathcal{R}^{\text{in}}_{s'} ,\zeta_{G'_{s'}})$ is supported by $\lambda^*_{s'}$. Thus criteria (2)-(4) hold.

    All that remains is to show that criterion (5) must hold. Since $G_s=G_{s'}$, the number of distinct vertices reached below by vertices in the bag $\mathcal{D}(s)$ given departure times $U$ must be equal to the number of vertices reached by vertices in the bag $\mathcal{D}(s')$ given departure times which account for the possibility of a path which traverses $u$. Since $u$ has no neighbours strictly below $\mathcal{D}(s')$, this is only possible if a path reaches a vertex in the bag by a time earlier than prescribed by $U$. By our assertions that criteria (2)-(4) hold, this departure time must be described by $\min (U,\mathcal{R}^{\text{in}}_s(u,U(u))$. For each function $U$ the vertices in $\mathcal{D}(s)$ must reach exactly $\mathcal{R}^{\text{below}}_{s'}(U')$ distinct vertices below the bag where $U'=\min (U,\mathcal{R}^{\text{in}}_s(u,U(u))$ by paths which depart at the times given by $U$ if $\mathcal{R}^{\text{below}}_{s'}(U')<k$ and at least $\mathcal{R}^{\text{below}}_{s'}(U')$ vertices otherwise. This follows from the fact that the subtrees rooted at $s$ and $s'$ consist of the same vertex set. Thus, criterion (5) must hold. Therefore the state $(P_{s},\mathcal{R}^{\text{below}}_{s}, \mathcal{R}^{\text{in}}_{s} ,\zeta_{G'_{s}})$ is valid if and only if the criteria hold.
\end{proof}

\subsubsection*{Forget nodes}
Suppose that a forget node $s$ has child $s'$ and that $\mathcal{D}(s')\setminus\mathcal{D}(s)=\{u\}$. 
\medskip
\begin{lemma}\label{lem:reach-nice-tree-forget}
   Let $s$ be a forget node with child $s'$ and forgotten vertex $u$. A state $(P_s,\mathcal{R}^{\text{below}}_s, \mathcal{R}^{\text{in}}_s ,\zeta_{G'_s})$ of a forget node $s$ is valid if and only if there exists a valid state $(P_{s'},\mathcal{R}^{\text{below}}_{s'}, \mathcal{R}^{\text{in}}_{s'} ,\zeta_{G'_{s'}})$ of $s'$ such that
\begin{enumerate}[leftmargin=0.5cm]
    \item $P_{s}=P_{s'}|_{E(G[\mathcal{D}(s')])}$;
    \item for each function $U\in \mathcal{U}_{\mathcal{D}(s)}$, if there is a vertex $w$ and finite times $t_1$, $t_2$ such that $U(w)=t_1$ and $\mathcal{R}^{\text{in}}_{s'}(w,t_1)(u)=t_2<\infty$, then $\mathcal{R}^{\text{below}}_s(U)=\min(\mathcal{R}^{\text{below}}_{s'}(U')+1,h)$ where $U'$ is the extension of $U$ where $U'(u)=t_2+1$, otherwise $\mathcal{R}^{\text{below}}_s(U)=\mathcal{R}^{\text{below}}_{s'}(U'')$ where $U''$ is the extension of $U$ such that $U''(u)=\infty$;
    \item $\mathcal{R}^{\text{in}}_s(v,t)=\mathcal{R}^{\text{in}}_{s'}(v,t)|_{\mathcal{D}(s)}$ for all $v\in \mathcal{D}(s)$;
    \item $\zeta_{G'_s}= \zeta_{G'_{s'}} + |\{(e,t):u\in e, (e,t)\in \text{changed}(P_{s'})\}|$. 
\end{enumerate}
\end{lemma}
\begin{proof}
    We show that we have a valid state $(P_s,\mathcal{R}^{\text{below}}_s, \mathcal{R}^{\text{in}}_s ,\zeta_{G'_s})$ of $s$ which fits the criteria given in the lemma statement if and only if there is a temporal assignment $\lambda'$ which supports it. 
    
    We begin by showing that if the criteria hold, the perturbation $\lambda'$ which supports the child state $(P_{s'},\mathcal{R}^{\text{below}}_{s'}, \mathcal{R}^{\text{in}}_{s'} ,\zeta_{G'_{s'}})$ must also support $(P_{s},\mathcal{R}^{\text{below}}_{s}, \mathcal{R}^{\text{in}}_{s} ,\zeta_{G'_{s}})$. Since we have a forget node, $G_s=G_{s'}$. Therefore, any $(\delta,\zeta)$-perturbation $\lambda'$ of $(G_{s'},\lambda)$ must be a $(\delta,\zeta)$-perturbation of $(G_{s},\lambda)$. By description of the state, the restriction of $P_{s'}$ to $P_s$ coincides with $\lambda'$ and the number of edges perturbed below $\mathcal{D}(s)$ must be exactly the sum of $\zeta_{G'_{s'}}$ and the edges which are perturbed under $P_{s'}$ which have been forgotten in $s$. Therefore, $P_{s'}$ and $\zeta_{G'_{s'}}$ must coincide with $\lambda'$. Since we allow paths to vertices in the bag which traverse forgotten vertices in our calculation of $\mathcal{R}^{\text{in}}_{s}$ and, by criterion (3), $\mathcal{R}^{\text{in}}_{s}=\mathcal{R}^{\text{in}}_{s'}|_{\mathcal{D}(s)}$, if $(P_{s'},\mathcal{R}^{\text{below}}_{s'}, \mathcal{R}^{\text{in}}_{s'} ,\zeta_{G'_{s'}})$ is supported by $\lambda'$, then our calculation of $\mathcal{R}^{\text{in}}_s$ must be supported by $\lambda'$. 
    
    Now suppose that $\mathcal{R}^{\text{below}}_s$ is as described in criterion (2). Then we show that the number of vertices below the bag $\mathcal{D}(s)$ temporally reachable under $\lambda'$ from vertices in $\mathcal{D}(s)$ after the times assigned by $U$ is exactly $\mathcal{R}^{\text{below}}_s(U)$ if $\mathcal{R}^{\text{below}}_s(U)<h$ and at least $\mathcal{R}^{\text{below}}_s(U)$ otherwise. We have two cases to consider: $\mathcal{R}^{\text{below}}_s(U)=\mathcal{R}^{\text{below}}_{s'}(U'')$ and $\mathcal{R}^{\text{below}}_s(U)=\min(\mathcal{R}^{\text{below}}_{s'}(U')+1,h)$. The latter occurs, by criterion (2), only if there is a vertex $w$ and finite times $t_1$, $t_2$ such that $U(w)=t_1$ and $\mathcal{R}^{\text{in}}_{s'}(w,t_1)(u)=t_2$. That is, it occurs if $u$ is temporally reachable from any vertex $w$ by a temporal path departing at $U(w)$. If this is true, the vertices in $\mathcal{D}(s)$ must reach all of the vertices below the bag reached by vertices of the child bag and $u$ by a path departing at either $u$ or another vertex $v\in\mathcal{D}(s)$ the time given by $U$. We note in particular that any vertices below the bag reached from $u$ must still be temporally reachable from vertices in the bag since the path from $w$ to $u$ arrives before the path from $u$ to vertices below the bag departs under $U'$. Therefore, in this case, the number of vertices below the bag $\mathcal{D}(s)$ temporally reachable from vertices in $\mathcal{D}(s)$ after the times assigned by $U$ is exactly $\mathcal{R}^{\text{below}}_s(U)$ under $\lambda'$ if $\mathcal{R}^{\text{below}}_s(U)<k$ and at least $\mathcal{R}^{\text{below}}_s(U)$ otherwise. 
    
    In the other case, there are no temporal paths from any vertex $v\in \mathcal{D}(s)$ departing at time $U(v)$ which traverse $u$ under $\lambda$. Therefore, the number of vertices below the bag $\mathcal{D}(s)$ temporally reachable from vertices in $\mathcal{D}(s)$ after the times assigned by $U$ is at least $\mathcal{R}^{\text{below}}_{s'}(U'')$ under $\lambda'$ where $U''$ is the extension of $U$ such that $U(u)=\infty$. 

    We now show the reverse direction, that if a state $(P_s,\mathcal{R}^{\text{below}}_s, \mathcal{R}^{\text{in}}_s ,\zeta_{G'_s})$ of $s$ is supported by some $(\delta,\zeta)$-perturbation $\lambda''$, then the criteria must hold. Since $G_s=G_{s'}$ and all $(\delta,\zeta)$-perturbations of $G_{s'}$ support valid states of $s'$, $\lambda''$ must support a valid state $(P_{s'},\mathcal{R}^{\text{below}}_{s'}, \mathcal{R}^{\text{in}}_{s'} ,\zeta_{G'_{s'}})$ of $s'$ such that $P_{s}=P_{s'}|_{E(G[\mathcal{D}(s')])}$ and $\zeta_{G'_s}= \zeta_{G'_{s'}} + |\{(e,t):u\in e, (e,t)\in \text{changed}(P_{s'})\}|$. What remains to show is that criteria (2) and (3) hold.

    We begin with criterion (3). Since $\lambda''$ also supports $(P_{s'},\mathcal{R}^{\text{below}}_{s'}, \mathcal{R}^{\text{in}}_{s'} ,\zeta_{G'_{s'}})$, the foremost arrival time of any vertex $v$ in $\mathcal{D}(s)$ under $\lambda''$ after some time $t$ must be exactly $\mathcal{R}^{\text{in}}_{s'}(v,t)$. Finally, the total number of distinct vertices below the bag $\mathcal{D}(s)$ reached by vertices $v$ in $\mathcal{D}(s)$ by paths departing at time $U(v)$ must precisely be the total number of vertices reached by vertices in $\mathcal{D}(s)$ departing at the times assigned by $U$ plus any vertices reached from $u$ departing after the foremost arrival of a path from a vertex in $\mathcal{D}(s)$. This is exactly the value of $\mathcal{R}^{\text{below}}_{s}$ as described by (2). Therefore, a state of a forget node is valid if and only if the criteria hold.
\end{proof}
\subsubsection*{Join nodes} 
Suppose that a join node $s$ has children $s_1,s_2$ and that $\mathcal{D}(s)=\mathcal{D}(s_1)=\mathcal{D}(s_2)$. 

We begin with the following definition which will aid us in calculating states of a join node. We use the observation that any foremost path in $G_s$ for a join node $s$ must alternately traverse paths consisting of vertices in the bag $\mathcal{D}(s)$ and vertices strictly below either child. We find the foremost arrival time at a vertex in the bag by splitting a path using $i$ ``alternations'' into two subpaths with some midpoint $x$ and checking whether there is a faster path traversing $x$ such that the paths to and from $x$ are each allowed fewer than $i$ alternations. The foremost arrival time using at most $i$ alternations is calculated recursively using the following formula.

\medskip
\begin{definition}\label{def:join-r-i}
    Let $(P_s,\mathcal{R}^{\text{below}}_s, \mathcal{R}^{\text{in}}_s ,\zeta_{G'_s})$ be a state of a join node $s$ with children $s_1$ and $s_2$. Let $(P_{s_1},\mathcal{R}^{\text{below}}_{s_1}, \mathcal{R}^{\text{in}}_{s_1} ,\zeta_{G'_{s_1}})$ be a valid state of $s_1$ and $(P_{s_2},\mathcal{R}^{\text{below}}_{s_2}, \mathcal{R}^{\text{in}}_{s_2} ,\zeta_{G'_{s_2}})$ be a valid state of $s_2$ such that $P_s=P_{s_1}=P_{s_2}$. For all pairs $(v,t)\in \mathcal{D}(s)\times [0,T(G,\lambda)+\delta]$, we define the function $\mathcal{R}^i:\mathcal{D}(s)\times[0,T(G,\lambda)+\delta]\to \mathcal{U}_{\mathcal{D}(s)}$ as follows: $\mathcal{R}^0(v,t)= \min (\mathcal{R}^{\text{in}}_{s_1}(v,t), \mathcal{R}^{\text{in}}_{s_2}(v,t))$; and for all $i>0$ 
    $$\mathcal{R}^i(v,t)=\min\left(\min_{x\in\mathcal{D}(s)}\left(\mathcal{R}^{i-1}(x,\mathcal{R}^{i-1}(v,t)(x)+1)\right),\mathcal{R}^{i-1}(v,t)\right).$$   
\end{definition}

We can now calculate the valid states of join nodes.
\medskip
\begin{lemma}\label{lem:reach-nice-tree-join}
    Let $s$ be a join node with children $s_1$ and $s_2$. A state $(P_s,\mathcal{R}^{\text{below}}_s, \mathcal{R}^{\text{in}}_s ,\zeta_{G'_s})$ of a join node $\mathcal{D}(s)$ is valid if and only if there exist valid states $(P_{s_1},\mathcal{R}^{\text{below}}_{s_1}, \mathcal{R}^{\text{in}}_{s_1} ,\zeta_{G'_{s_1}})$ of $s_1$ and $(P_{s_2},\mathcal{R}^{\text{below}}_{s_2}, \mathcal{R}^{\text{in}}_{s_2} ,\zeta_{G'_{s_2}})$ of $s_2$ such that
\begin{enumerate}[leftmargin=0.5cm]
    \item $P_{s_1}=P_{s_2}=P_s$;
    \item $\mathcal{R}^{\text{in}}_s= \mathcal{R}^{\lceil\log|\mathcal{D}(s)|\rceil}$ (as defined in Definition~\ref{def:join-r-i});
    \item for all functions $U\in\mathcal{U}_{\mathcal{D}(s)}$, $\mathcal{R}^{\text{below}}_{s}(U)=\min(\mathcal{R}^{\text{below}}_{s_1}(U')+\mathcal{R}^{\text{below}}_{s_2}(U'),h)$ where $U'(v)=\min(U(v),\min_{w\in\mathcal{D}(s)}\mathcal{R}^{\text{in}}_s(w,U(w))(v)+1)$ for all $v\in\mathcal{D}(s)$;
    \item $\zeta_{G'_s}=\zeta_{G'_{s_1}}+\zeta_{G'_{s_2}}$
    \item $\zeta_{G_s}+|\text{changed}(P_s)|\leq \zeta$.
\end{enumerate}
\end{lemma}
\begin{proof}
    We claim that a state $(P_s,\mathcal{R}^{\text{below}}_s, \mathcal{R}^{\text{in}}_s ,\zeta_{G'_s})$ is supported by a $(\delta,\zeta)$-perturbation $\lambda'$ of $(G,\lambda)$ if and only if restrictions of $\lambda'$ support the states $(P_{s_1},\mathcal{R}^{\text{below}}_{s_1}, \mathcal{R}^{\text{in}}_{s_1} ,\zeta_{G'_{s_1}})$ and $(P_{s_2},\mathcal{R}^{\text{below}}_{s_2}, \mathcal{R}^{\text{in}}_{s_2} ,\zeta_{G'_{s_2}})$ of $s_1$ and $s_2$ respectively and the criteria in the lemma hold. We begin by showing that the state $(P_{s},\mathcal{R}^{\text{below}}_{s}, \mathcal{R}^{\text{in}}_{s} ,\zeta_{G'_{s}})$ is supported by some $(\delta,\zeta)$-perturbation $\lambda'$ if the criteria hold. We refer to the $(\delta,\zeta)$-perturbations supporting $s_1$ and $s_2$ as $\lambda_1$ and $\lambda_2$ respectively. Let $\lambda'$ be the $(\delta,\zeta)$-perturbation of $G_s$ which is equal to $\lambda_1$ when restricted to $G_{s_1}$ and equal to $\lambda_2$ when restricted to $G_{s_2}$. This is well-defined by criterion (1).

    Since the set of vertices in the bags of nodes $s$, $s_1$ and $s_2$ is the same, and $P_s=P_{s_1}=P_{s_2}$, $\lambda'$ must coincide with $P_s$ when restricted to the edges in $G[\mathcal{D}(s)]$. In addition, since the edge sets of $G'_{s_1}$ and $G'_{s_2}$ are disjoint, $\lambda'$ must perturb exactly $\zeta_{G_s}$ edges below the bag. By criterion (5), the number of edges perturbed by $\lambda$ in total must be at most $\zeta$. Thus $\lambda'$ is a $(\delta,\zeta)$-perturbation of $G_s$. We now show that if criteria (2) and (3) hold, $\lambda'$ supports the state $(P_s,\mathcal{R}^{\text{below}}_s, \mathcal{R}^{\text{in}}_s ,\zeta_{G'_s})$.
    \begin{claim}\label{claim:join}
         A foremost path in $(G_s,\lambda')$ from $v$ to $u$ departing at $t$ arrives at $t'$ for all times $t,t'$ and vertices $u,v$ in $\mathcal{D}(s)$ if and only if $\mathcal{R}^{\text{in}}_s(v,t)(u)=t'$.
    \end{claim}
    \begin{proof}
    
    We begin with the observation that any foremost temporal path between two vertices in $(G_s,\lambda')$ must be a concatenation of (possibly trivial) paths alternating between $G_{s_1}$ and $G_{s_2}$ which are separated by paths consisting of at least one vertex in $\mathcal{D}(s)$. This is because the only edges in $G_s$ with an end point in each of $G_{s_1}$ and $G_{s_2}$ must be in $G[\mathcal{D}(s)]$. Since there are at most $|\mathcal{D}(s)|$ distinct vertices in $\mathcal{D}(s)$, there are at most $|\mathcal{D}(s)|$ subpaths of a foremost path contained entirely in $G[\mathcal{D}(s)]$. We will refer to a maximal subpath contained entirely in $G[\mathcal{D}(s)]$ as a \emph{separating} subpath.

    We will show by induction on $i$ that $\mathcal{R}^i(v,t)(w)=t'$ if and only if $t'$ is the earliest arrival time of a path from $v$ to $x$ departing at $t$ which contains at most $2^i+1$ separating subpaths. We begin with our base case, $i=0$. Here we set the value of $\mathcal{R}^0(v,t)(w)$ to be $\min (\mathcal{R}^{\text{in}}_{s_1}(v,t)(w), \mathcal{R}^{\text{in}}_{s_2}(v,t)(w))$ for all $v,w\in\mathcal{D}(s)$ and $t\in [0,T(G,\lambda)+\delta]$. We first show that if a path containing at most two separating subpaths from $v$ to $w$ departing at time $t$ arrives at time $t'$, then $\mathcal{R}^0(v,t)(w)=t'$. Note that since we are looking for a path from $v\in\mathcal{D}(s)$ to $w\in \mathcal{D}(s)$ there must be at least one separating subpath if any temporal path exists. Therefore, we must show that $\min (\mathcal{R}^{\text{in}}_{s_1}(v,t)(w), \mathcal{R}^{\text{in}}_{s_2}(v,t)(w))$ for all $v,w\in\mathcal{D}(s)$ is the earliest arrival time of a temporal path containing either 1 or 2 separating subpaths.
    If a path contains at most 2 separating subpaths, then it must contain at most one subpath which traverses vertices in only one of $G'_{s_1}$ or $G'_{s_2}$. Assume without loss of generality, that it traverses edges in the former subgraph. Then a path from $v$ to $w$ with at most 2 separating subpaths in $G_s$ departing at time $t$ must be of the form $p_v p_1 p_w$ where $p_v$ and $p_w$ are separating subpaths containing $v$ and $w$ respectively and $p_1$ is a path traversing only vertices in $V(G_{s_1})\setminus \mathcal{D}(s)$. By validity of $(P_{s_1},\mathcal{R}^{\text{below}}_{s_1}, \mathcal{R}^{\text{in}}_{s_1} ,\zeta_{G'_{s_1}})$, the arrival time of this path must be $\mathcal{R}^{\text{in}}_{s_1}(v,t)(w)$ which is at most $\mathcal{R}^{\text{in}}_{s_2}(v,t)(w)$. Thus $\mathcal{R}^0(v,t)(w)$ gives the earliest arrival time of a path with at most 2 separating subpaths from $v$ to $w$ departing at time $t$. Thus, $\mathcal{R}^0(v,t)(w)=t'$ if and only if the earliest arrival time of a path containing at most two separating subpaths from $v$ to $w$ departing at time $t$ is $t'$ under $\lambda'$.

    We now assume that for all $j\leq i$, $\mathcal{R}^j(v,t)(w)=t'$ if and only if the earliest arrival time of a path with at most $2^j+1$ separating subpaths from $v$ to $w$ departing at time $t$ in $(G_s,\lambda')$ arrives at $t'$. We begin by showing that if $\mathcal{R}^{i+1}(v,t)(w)=t'$, then a path containing at most $2^{i+1}+1$ separating subpaths from $v$ to $w$ departing at time $t$ which arrives earliest arrives at time at least $t'$. For all vertices $v$, $w\in\mathcal{D}(s)$ and times $t$, we calculate $\mathcal{R}^{i+1}(v,t)(w)$ by finding the minimum of $\mathcal{R}^{i}(v,t)(w)$ and $\min_{x\in\mathcal{D}(s)}\mathcal{R}^{i}(x,\mathcal{R}^{i}(v,t)(x)+1)(w)$ for all $x\in\mathcal{D}(s)$. Suppose, for a contradiction, that under $\lambda'$ the earliest arrival time of a path $P$ containing at most $2^{i+1}+1$ separating subpaths from $v$ to $w$ departing at $t$ is $t''<t'$. We have two cases to consider. Namely, $P$ contains at most $2^i+1$ separating subpaths, or the number of separating subpaths in $P$ is between $2^j+1$ and $2^{i+1}+1$. If the former is true, then $\mathcal{R}^{i}(v,t)(w)=t''$ by the inductive hypothesis; a contradiction.
    
    Now suppose that the number of separating subpaths in $P$ is between $2^i+1$ and $2^{i+1}+1$. Let $x$ be the final vertex in the $(2^i)$th separating subpath, $P_{vx}$ be the subpath of $P$ from $v$ to $x$, and $P_{xw}$ be the subpath of $P$ from $x$ to $w$. Then, $P_{vx}$ and $P_{xw}$ must contain at most $2^{i}+1$ and at least 2 separating subpaths each. In addition, $P_{vx}$ must be a path containing at most $2^i+1$ separating paths from $v$ to $x$ departing at time $t$ with earliest arrival time, and $P_{xw}$ must be a path containing at most $2^i+1$ separating paths from $x$ to $w$ departing strictly after the arrival time of $P_{vx}$ with earliest arrival time.
    Thus, $t''\geq\mathcal{R}^{i}(x,\mathcal{R}^{i}(v,t)(x)+1)(w)$, a contradiction of our assumption that $t''<t'$. 
    
    We now show that if $\mathcal{R}^{i+1}(v,t)(w)=t'$, then a path containing at most $2^{i+1}+1$ separating subpaths from $v$ to $w$ departing at time $t$ which arrives earliest arrives at time at most $t'$. If $\mathcal{R}^{i+1}(v,t)(w)=t'$, either $t'=\mathcal{R}^{i-1}(v,t)$ or there exists a vertex $x\in\mathcal{D}(s)$ such that 
    $t'=\mathcal{R}^{i-1}(x,\mathcal{R}^{i-1}(v,t)(x)+1)$. In both cases, this implies there exists a path containing at most $2^{i+1}+1$ separating subpaths from $v$ to $w$ departing at time $t$ which arrives at time $t'$. Thus the earliest such path must arrive at time $t''\geq \mathcal{R}^{i+1}(v,t)(w)$.
    Therefore, if $\mathcal{R}^{i+1}(v,t)(w)=t'$, then a path containing at most $2^{i+1}+1$ separating subpaths that arrives earliest from $v$ to $w$ departing at time $t$ arrives at time $t'$.

    We now show that if the path containing at most $2^{i+1}+1$ separating subpaths from $v$ to $w$ departing at time $t$ which arrives earliest arrives at time $t'$, then $\mathcal{R}^{i+1}(v,t)(w)=t'$. We note that a path containing at most $2^{i+1}+1$ separating subpaths must either contain at most $2^i+1$ separating subpaths, or be the concatenation of 2 subpaths consisting of a path from $v$ to a vertex $x\in\mathcal{D}(s)$ departing at time $t$ containing at most $2^i+1$ separating subpaths and a path from $x$ to $w$ departing at time $t'+1$ containing at most $2^i+1$ separating subpaths where $t'$ is the time of arrival of the former subpath. Therefore, we can find the earliest arrival time of such a path by finding the vertex $x$ in $\mathcal{D}(s)$ that minimises the earliest arrival time of a path comprised of a subpath containing at most $2^i+1$ separating subpaths from $v$ to $x$ and a subpath containing at most $2^i+1$ separating subpaths from $x$ to $w$, and taking the minimum of this arrival time and the minimum time of arrival of a subpath containing at most $2^i+1$ separating subpaths. This is precisely our calculation of $\mathcal{R}^{i+1}$. Thus, $\mathcal{R}^{i+1}(v,t)(w)=t'$ if and only if the earliest arrival time of a path with at most $2^{i+1}+1$ separating subpaths from $v$ to $w$ departing at time $t$ in $(G_s,\lambda')$ arrives at $t'$. By induction, $\mathcal{R}^i(v,t)(w)=t'$ if and only if the earliest arrival time of a path with at most $2^i+1$ separating subpaths from $v$ to $w$ departing at time $t$ in $(G_s,\lambda')$ arrives at $t'$ for all $i\in\mathbb{N}$.

    As stated earlier, there are at most $|\mathcal{D}(s)|$ subpaths of a foremost path contained entirely in $G[\mathcal{D}(s)]$ since there are at most $|\mathcal{D}(s)|$ distinct vertices in $\mathcal{D}(s)$. Thus, there are at most $|\mathcal{D}(s)|$ separating paths in a foremost path in $G_s$. Therefore, the foremost arrival time under $\lambda'$ of a path given a departing vertex and time must be given by $\mathcal{R}^{\lceil\log|\mathcal{D}(s)|\rceil}$. Therefore if (2) holds, $\mathcal{R}^{\text{in}}_s$ matches the state supported by $\lambda'$. 
    \end{proof}

    We now show that, if (3) holds, the state $(P_s,\mathcal{R}^{\text{below}}_s, \mathcal{R}^{\text{in}}_s ,\zeta_{G'_s})$ must be supported by $\lambda'$. We do this by showing that, under $\lambda'$, the number of distinct vertices reached below the bag by vertices in the bag at the departure times given by any function $U$ is precisely $\mathcal{R}^{\text{below}}_s$ as given by (3). The departure time of the paths reaching vertices below the bag must be after the foremost arrival time of a path from a vertex in the bag departing at the time specified by $U$. This includes the trivial path from $v$ to itself departing at $U(v)$. Note that, since each vertex is potentially assigned a different departure time by $U$, a path could arrive from a vertex $w$ in $\mathcal{D}(s)$ at a vertex $v\in \mathcal{D}(s)$ before $U(v)$ by a path departing at $U(w)$.
    Since the sets of vertices in $G_{s_1}\setminus\mathcal{D}(s)$ and $G_{s_2}\setminus\mathcal{D}(s)$ are disjoint, the number of vertices below the bag which are temporally reachable from a vertex in the bag at the departure times given under $\lambda'$ must be the sum of those reachable from that vertex in $G_{s_1}\setminus\mathcal{D}(s)$ and $G_{s_2}\setminus\mathcal{D}(s)$ under the restriction of $\lambda$ to each subgraph. This is precisely the number described by (3) (upper bounded by $h$). Therefore, if all criteria hold, the $(\delta,\zeta)$-perturbation of $G_s$ which is equal to $\lambda_1$ when restricted to $G_{s_1}$ and equal to $\lambda_2$ when restricted to $G_{s_2}$ must support the state $(P_s,\mathcal{R}^{\text{below}}_s, \mathcal{R}^{\text{in}}_s ,\zeta_{G'_s})$.

    Now suppose we have a $(\delta,\zeta)$-perturbation of $G_s$. Call this perturbation $\lambda''$. We will show that the criteria must hold. We denote by $\lambda''_1$ and $\lambda''_2$ the restrictions of $\lambda''$ to $G_{s_1}$ and $G_{s_2}$ respectively. By construction, criteria (1), (4), and (5) must hold. Furthermore, given states supported by each of $\lambda''_1$ and $\lambda''_2$, the foremost arrival time of paths between vertices in the bag under $\lambda''$ must be the foremost arrival time of a path containing at most $\mathcal{D}(s)$ separating subpaths. As shown in Claim~\ref{claim:join}, this is precisely $\mathcal{R}^{\lceil\log|\mathcal{D}(s)|\rceil}=\mathcal{R}^{\text{in}}_s$. Therefore, criterion (2) must hold. Finally, for any function $U\in \mathcal{U}_{\mathcal{D}(s)}$, the number of distinct vertices reachable from vertices $v\in\mathcal{D}(s)$ departing at $U(v)$ under $\lambda''$ must be the sum of vertices reachable strictly below the bags $\mathcal{D}(s_1)$ and $\mathcal{D}(s_2)$ by paths departing after the earliest time of arrival at vertices in the bag by paths departing at the times given by $U$. Hence, criterion (3) gives $\mathcal{R}^{\text{below}}_s$ under $\lambda''$. Thus all criteria hold and a state $(P_s,\mathcal{R}^{\text{below}}_s, \mathcal{R}^{\text{in}}_s ,\zeta_{G'_s})$ is supported by a $(\delta,\zeta)$-perturbation $\lambda'$ of $(G,\lambda)$ if and only if restrictions of $\lambda'$ support the states $(P_{s_1},\mathcal{R}^{\text{below}}_{s_1}, \mathcal{R}^{\text{in}}_{s_1} ,\zeta_{G'_{s_1}})$ and $(P_{s_2},\mathcal{R}^{\text{below}}_{s_2}, \mathcal{R}^{\text{in}}_{s_2} ,\zeta_{G'_{s_2}})$ of $s_1$ and $s_2$ respectively and the criteria in the lemma hold.
\end{proof}
\subsubsection{Proof of Theorem \ref{thm:trlp-tree-decomp}}
\begin{theorem*}[Theorem \ref{thm:trlp-tree-decomp} restated]
    Given a temporal graph $(G,\lambda)$ where the treewidth of $G$ is at most $\omega$, we can solve TRLP in time
    $$O \left( n\omega^5 \zeta^3 (T(G,\lambda)+\delta)^{\omega})  \cdot (2 \delta h  T(G,\lambda) )^{4\omega^2 (T(G,\lambda) + \delta)^\omega}   \right).$$
    In particular, we can solve TRLP in polynomial time when the treewidth of the underlying graph, the permitted perturbation size $\delta$ and the lifetime of $(G,\lambda)$ are all bounded by constants.
\end{theorem*}
\begin{proof}
We can compute a tree decomposition in $\omega^{O(\omega^3)}n$ time given a graph on $n$ vertices with treewidth $\omega$ \cite{bodlaender_linear-time_1996}. Furthermore, given a tree decomposition of a graph of width at most $\omega$, we can find a nice tree decomposition of the graph of width at most $\omega$ in time bounded by $O(\omega^3n)$ \cite{Cygan2015}. Given a nice tree decomposition of a temporal graph $(G,\lambda)$ as described in Definition~\ref{def:nice-tree-decomp}, we solve TRLP by finding all valid states of the root node for each choice of source vertex $x$. For such a node, we check if there exists a valid state of the root node $s$ containing a source vertex $x$ such that $\mathcal{R}^{\text{below}}_s(U)\geq h-1$ where $U(x)=0$.

Our algorithm works up from the leaves to the root, computing all valid states for each node, so that when we consider any node we have access to the set of valid states for all its children. We use Lemmas~\ref{lem:reach-nice-tree-intro},~\ref{lem:reach-nice-tree-forget} and~\ref{lem:reach-nice-tree-join} to find the valid states of each type of node given the valid states of its children. To bound the running time of the algorithm, we bound the time needed to compute the valid states of a node given the valid states of its children (if there are any). 

We can bound the number of possible states for each bag by considering the number of possibilities for each component of a state. There are at most $\zeta$ edges perturbed below a given bag. The number of time-edges per bag is bounded above by $O(w^2\cdot T(G,\lambda))$ where $w$ is the treewidth of the underlying graph $G$. For a given time-edge, there are $2\delta$ times that it could be active under a perturbation of $\lambda$, therefore there are at most $(2\delta)^{w^2T(G,\lambda)}$ possible functions $P_s$. Recall that $\mathcal{R}^{\text{in}}$ is a function from the set $\mathcal{D}(s)\times[0,T(G,\lambda)+\delta]$ to the set $\mathcal{U}_{\mathcal{D}(s)}$. There are at most $\omega (T(G,\lambda)+\delta)$ and $(T(G,\lambda)+\delta+1)^{\omega})$ elements in the domain and codomain of $\mathcal{R}^{\text{in}}$ respectively. Thus, there are at most $(T(G,\lambda)+\delta)^{\omega^2(T(G,\lambda)+\delta)}$ functions $\mathcal{R}^{\text{in}}$ for each bag $\mathcal{D}(s)$.

As stated earlier, for a given bag $\mathcal{D}(s)$, there are at most $(T(G,\lambda)+\delta+1)^{\omega}$ functions in $\mathcal{U}_{\mathcal{D}(s)}$. Since this is the domain of $\mathcal{R}^{\text{below}}$ and the codomain is of cardinality $h+1$, there are at most
$$h^{(T(G,\lambda)+\delta+1)^{{\omega}}}$$
possible functions $\mathcal{R}^{\text{below}}$.
There are therefore at most 
\begin{align*}
    O\left(\zeta\cdot (2\delta)^{w^2T(G,\lambda)}\cdot (T(G,\lambda)+\delta+1)^{\omega^2(T(G,\lambda)+\delta)}\cdot h^{(T(G,\lambda)+\delta)^{{\omega}}}\right) \\
    = O\left( \zeta \cdot (2 \delta h  T(G,\lambda) )^{\omega^2 (T(G,\lambda) + \delta)^\omega} \right)
\end{align*}
possible states for a node in the nice tree decomposition. 

We can find the valid states for leaf nodes in constant time. Given a state of an introduce node, we check whether it is valid by checking that each of the 7 criteria in Lemma~\ref{lem:reach-nice-tree-intro} hold for a valid state of its child. Given a child state, the time needed to check each criterion is as follows:
\begin{itemize}[leftmargin=0.5cm]
    \item For Criterion 1, we can check in time $O(\omega^2\tau)$ that $P_s$ extends $P_{s'}$ of the child state by examining each time-edge in the bag.
    \item For each of Criteria 2-4, and for each pair $(v,t)$, we need time $O(\omega \tau)$ to consider all time-edges incident with $u$ and perform constant-time lookups for each; repeating for every pair $(v,t)$ means the total time needed to check all three criteria is $O(\omega^2 \tau (T(G,\lambda)+\delta))$.
    \item For Criterion 5, for a fixed $U$ we need time $O(\omega)$ to compute the updated departure times for all vertices, and then look up corresponding values in the child state; summing over all possible functions $U$ means this criterion can be checked in time $O(\omega(T(G,\lambda)+\delta)^{\omega})$.
    \item For Criterion 6 we need only constant time to check that $\zeta_{G'_s}=\zeta_{G'_{s'}}$.
    \item For Criterion 7, we can compute the number of perturbed time-edges under $P_s$ in time $O(\omega^2 \tau)$ by considering every time-edge in the bag, then perform a constant-time arithmetic operation.
\end{itemize}
Overall, this means the time needed to check consistency of a given combination of states for a parent and child is
$O(\omega^2(T(G,\lambda)+\delta)^{\omega})$.  This must be repeated for every pair of possible parent and child states, giving us a run-time of 
$$O\left( \zeta^2\omega^2 \cdot (2 \delta h  T(G,\lambda) )^{3\omega^2 (T(G,\lambda) + \delta)^\omega} \right) .$$

To find valid states of forget nodes we check consistency of each combination of parent and child states using the criteria of Lemma \ref{lem:reach-nice-tree-forget} as follows:
\begin{itemize}[leftmargin=0.5cm]
    \item For Criterion 1, we can check in time $O(\omega^2\tau)$ that $P_s$ is a restriction of $P_{s'}$ of the child state by examining each time-edge in the bag.
    \item To check Criterion 2, for each function $U \in \mathcal{U}_{\mathcal{D}(s)}$ we must perform a constant number of constant-time operations for each vertex $w \in \mathcal{D}(s)$; repeating for every $U$ and $w$ therefore takes time $O(\omega(T(G,\lambda)+\delta)^{\omega})$.
    \item For Criterion 3, for each pair $(v,t)$  with $v \in \mathcal{D}(s)$ we check that $\mathcal{R}^{\text{in}}_s(v,t) = \mathcal{R}^{\text{in}}_{s'}$ in time $O(\omega)$; repeating this for each pair $(v,t)$ gives a total time of $O(\omega^2(T(G,\lambda)+\delta))$ to check the criterion.
    \item We require $O(\omega\tau)$ to check Criterion 4 by counting the perturbed edges incident to the forgotten vertex.
\end{itemize}
Overall, this means the time needed to check consistency of a given combination of states for a parent and child is $O(\omega^2(T(G,\lambda) + \delta)^\omega)$. This must be repeated for every pair of possible parent and child states, giving us a run-time of 
$$O\left( \zeta^2\omega^2 \cdot (2 \delta h  T(G,\lambda) )^{3\omega^2 (T(G,\lambda) + \delta)^\omega} \right) .$$

For join nodes we need to check combinations of a parent state and a state for each child.  The time needed to check whether the criteria of Lemma~\ref{lem:reach-nice-tree-join} hold for each combination is as follows:
\begin{itemize}[leftmargin=0.5cm]
    \item For Criterion 1, we can check in time $O(\omega^2\tau)$ that $P_s=P_{s_1}=P_{s_2}$ by examining each time-edge in the bag.
    \item To check Criterion 2 we must compute $R^{\lceil \log \omega \rceil}$ as in Definition~\ref{def:join-r-i}.  Initialising $R^0(v,t)(w)$ takes constant time for each pair $(v,t)$ and vertex $w$, so time $O(\omega^2(T(G,\lambda) + \delta))$ overall.  To compute $\mathcal{R}^i(v,t)$ for a fixed pair $(v,t)$ and vertex $w$, given constant-time access to all values of $\mathcal{R}^{i-1}$, we require time $O(\omega)$; Calculating this for all pairs $(v,t)$ and vertices $w$ therefore takes time $O(\omega^3(T(G,\lambda) + \delta))$.  Since we repeat this process $O(\log \omega)$ times, the total time needed to check Criterion 2 is $O(\omega^3\log \omega(T(G,\lambda)+\delta))$ time.
    \item Criterion 3 requires us to compute an updated function for the departure times of the paths. For each vertex in the bag $\mathcal{D}(s)$, this takes time $O(\omega)$. Checking the criterion for all functions $U\in\mathcal{U}_{\mathcal{D}(s)}$, therefore requires $O(\omega^2(T(G,\lambda)+\delta)^{\omega})$ time.
    \item We can check Criterion 4 in constant time by checking that $\zeta_{G_s}$ is the sum of $\zeta_{G_{s_1}}$ and $\zeta_{G_{s_2}}$.
    \item Criterion 5, we can compute the number of perturbed time-edges under $P_s$ in time $O(\omega^2 \tau)$ by considering every time-edge in the bag, then perform a constant-time arithmetic operation.
\end{itemize}
Overall, this means that the time needed to check consistency of a given combination of states for a parent and its two children is $O(\omega^3\log\omega (T(G,\lambda)+\delta)^{\omega}))$.  This must be repeated for every triple of possible parent and child states, giving us a run-time of 
$$O \left( \omega^4 \zeta^3 (T(G,\lambda)+\delta)^{\omega})  \cdot (2 \delta h  T(G,\lambda) )^{4\omega^2 (T(G,\lambda) + \delta)^\omega}   \right).$$



To solve TRLP, we must find all valid states of every node in the nice tree decomposition. Recall that we have $O(\omega n)$ nodes in a nice tree decomposition, where $n$ is the number of vertices in the input graph. Finding valid states of join nodes is the most time-consuming part of the algorithm, therefore the worst-case running time is 
$$O \left( n\omega^5 \zeta^3 (T(G,\lambda)+\delta)^{\omega})  \cdot (2 \delta h  T(G,\lambda) )^{4\omega^2 (T(G,\lambda) + \delta)^\omega}   \right).$$
\end{proof}

\section{Eccentricity}\label{sec:eccentricity}
We now consider two closely related problems under the same model of perturbation.  Instead of asking about the number of reachable vertices, we are interested in whether it is possible to reach \emph{all} vertices in the graph by paths whose length or duration is at most some upper bound.  It is unsurprising that the problem of determining whether there exists a set of perturbations achieving either goal is NP-complete; in contrast to TRLP, however, it turns out that these problems do not become easier when we allow an appearance of every edge to be perturbed.  If all these perturbations can also be sufficiently large, however, we do still recover tractability.

We have previously discussed strict temporal paths, as well as their lengths. We will use the terms \emph{shortest} and \emph{longest} to be specifically in relation to the number of edges in the path.
However, there are other measures of a temporal path that may be relevant under various real-world scenarios. Earlier, we defined foremost paths and the arrival time of a foremost path. In addition, we say that the \emph{duration} of a path $P:= ((v_0v_1,t_1),\ldots, (v_{\ell-1}v_\ell,t_\ell))$ is $t_\ell-t_1$.
A path $P$ is \emph{fastest} of a set of paths $\mathcal{P}$ if there is no other path $P'\in \mathcal{P}$ of strictly smaller duration.

We begin with formal definitions of the problems we are considering.  In both definitions, we are given a temporal graph $(G,\lambda)$ and a source vertex $v_s$.
This vertex $v_s$ has a \emph{temporal shortest eccentricity} of at most $k$ if, for every vertex $v_t\in V(G)$, a shortest temporal path from $v_s$ to $v_t$ contains at most $k$ edges.
The vertex $v_s$ has a \emph{temporal fastest eccentricity} of at most $k$ if, for every vertex $v_t\in V(G)$, a fastest temporal path from $v_s$ to $v_t$ has a duration of at most $k$.
We use the notation $\teccs((G,\lambda),v)$ (respectively $\teccfa((G,\lambda),v)$) to mean the temporal shortest (respectively fastest) eccentricity of $v$ in the temporal graph $(G,\lambda)$.  We note that, given a temporal graph $(G,\lambda)$ and a vertex $v \in V(G)$, both $\teccs((G,\lambda),v)$ and $\teccfa((G,\lambda),v)$ can be computed in polynomial time \cite{XUAN2003}.

We now define the problems of determining whether there is a perturbation achieving some specified eccentricity.

\begin{framed}
\noindent
\textbf{\textsc{Temporal shortest eccentricity under perturbation (TSEP)}}\\
\textit{Input:} A temporal graph $(G,\lambda)$, a source vertex $v_s \in V(G)$, and integers $k$, $\delta$, and $\zeta$.\\
\textit{Question:} Is there a $(\delta,\zeta)$-perturbation $\lambda'$ of $\lambda$ such that $\teccs((G,\lambda'),v_s) \leq k$?
\end{framed}

\begin{framed}
\noindent
\textbf{\textsc{Temporal fastest eccentricity under perturbation (TFaEP)}}\\
\textit{Input:} A temporal graph $(G,\lambda)$, a source vertex $v_s \in V(G)$, and integers $k$, $\delta$, and $\zeta$.\\
\textit{Question:} Is there a $(\delta,\zeta)$-perturbation $\lambda'$ of $\lambda$ such that $\teccfa((G,\lambda'),v_s) \leq k$?
\end{framed}

Here we have considered two of the three standard notions of optimality for temporal paths.  One could naturally define a third problem based on temporal \emph{foremost} eccentricity (requiring the earliest arrival time of a path from $v$ to every other vertex to be at most $k$), but this problem has previously been studied by Deligkas et al.~\cite{deligkas_minimizing_2023_bugfree} under another name.  The problem \textsc{ReachFast}$(k)$ is to minimise the foremost eccentricity of some \emph{set} $S$ of sources (i.e.~minimise the maximum eccentricity over all $s \in S$) when at most $k$ edges can be perturbed, while \textsc{ReachFastTotal}$(k)$ asks the same question when the sum of the sizes of all perturbations must be bounded by $k$.  The behaviour of these problems is similar to what we see for TSEP and TFaEP: they are intractable in general, but can be solved in polynomial time if both the number and size of the perturbations are unrestricted.




We begin by observing that TSEP and TFaEP are both NP-complete.  Inclusion in NP is trivial (the set of perturbations acts as a certificate), and to see NP-hardness we note that in the proof of \Cref{lemma:ds-to-reachability} our constructed temporal graph contains a source vertex that reaches all vertices via a path of at most two edges and of duration at most two if and only if the original graph has a dominating set of size $r$.  (We note that Construction~\ref{const:domset} is similar to the construction used by Deligkas et al.~\cite[Theorem 2]{deligkas_minimizing_2023_bugfree} to show hardness of \textsc{ReachFast}$(k)$ and \textsc{ReachFastTotal}$(k)$ by reduction from \textsc{Hitting Set}.)  We can further conclude from this observation (as in \Cref{thm:inapprox-reach}) that both problems are W[2]-hard parameterised by $\zeta$ and are NP-hard to approximate (where the goal is to find the smallest number of perturbations required).

Recall that TRLP became polynomial-time solvable when we were allowed to perturb at least one appearance of every edge.  In contrast with this result, we now show that both TFaEP and TSEP remain NP-hard even when an appearance of every edge can be perturbed, for any constant $\delta$, and for constant values of $k$. 
Note that this is also a contrast to the behaviour for foremost eccentricity: Deligkas et al.~\cite[Algorithm 1]{deligkas_minimizing_2023_bugfree} show that this problem is efficiently solvable when $\zeta$ is large.

We begin with TSEP, and give a hardness result via a reduction from SAT.
\medskip
\begin{theorem}\label{theorem:tsep-many-perturbed-np-hard}
     For any integers $k \geq 4$ and $\delta\geq 1$, TSEP is \NP-hard even if $\zeta=m$.
\end{theorem}
\begin{proof}
    We show how, given integers $\delta$ and $k\geq 4$, an instance $((G,\lambda), k, 1, m)$ (where $m=|E(G)|$) of TSEP can be constructed from any instance $\bigwedge_j C_j$ of SAT
    such that $((G,\lambda),k, \delta, m)$ is a yes-instance of TSEP if and only if $C\bigwedge_j C_j$ is a yes-instance of SAT.
    To begin, we construct the temporal graph $(G,\lambda)$ from an instance $\bigwedge_j C_j$ of SAT on variables $\{x_i\}$ as follows.
    First, create the vertex $v_s$.
    Then, for each clause $C_j$ create the vertex $v^c_j$ (collectively referred to as clause vertices), and
    for each variable $x_i$ define $v_{i,0} = v_s$,
    create the vertices $v_{i,\ell}$ for $\ell\in[1,k]$,
    for $\ell\in[1,k-4]$ add the temporal edges $(v_{i,\ell}v_{i,\ell+1}, \ell+1)$,
    and add the temporal edges $(v_{i, k-4}v_{i, k-3},k-3)$, $(v_{i,k-4}v_{i,k-2},3\delta+k-2)$, $(v_{i,k-3}v_{i,k-2},k-2)$, $(v_{i,k-2}v_{i,k-1}, \delta+k-1)$, and $(v_{i,k-1}v_{i,k},3\delta+k)$.
    Then, for each $C_j$ that contains $x_i$ as a literal we add the temporal edge $(v_{i,k-1}v^c_j,k)$ and for each $C_j$ that contains $\overline x_i$ as a literal we add the temporal edge $(v_{i,k}v^c_j, 3\delta+k+1)$.
    A diagram of this appears in Figure~\ref{fig:tsep-many-perturbed-np-hard}.

    We now note some properties of this construction.
    First, $v_{i,k-4}v_{i,k-2}v_{i,k-1}v^c_j$ can never form part of a temporal path in $(G,\lambda')$ for any $\delta$-perturbation $\lambda'$ of $\lambda$
    as both $\lambda'(v_{i,k-4}v_{i,k-2}) \geq 2\delta+k-2$ and  $\lambda'(v_{i,k-1}v^c_j) \leq \delta+k$ hold for every $\delta$-perturbation.
    Therefore any path of length $k$ from $v_s$ to an arbitrary clause vertex $v^c_j$ in $(G,\lambda')$ for some $\delta$-perturbation $\lambda'$ must be either of the form $v_sv_{i,1}v_{i,2}\ldots v_{i,k-3}v_{i,k-2}v_{i,k-1}v^c_j$ for some $i$ (we call this a $\top$-path for $x_i$), or of the form $v_sv_{i,1}v_{i,2}\ldots v_{i,k-4}v_{i,k-2}v_{i,k-1}v_{i,k}v^c_j$ for some variable $x_i$ (we will call this a $\bot$-path for $x_i$).
    If we have a $\top$-path for $x_i$ in $(G,\lambda')$ then $\lambda'(v_{i,k-2}v_{i,k-1})\leq \delta+k-1$, and if we have a $\bot$-path for $x_i$ in $(G,\lambda')$ then $\lambda'(v_{i,k-2}v_{i,k-1} \geq 2\delta+k-1$.

    Next we show that $((G,\lambda), v_s, k, 1, m)$ is a yes-instance for TSEP if and only if $\bigwedge_j C_j$ is satisfiable.
    First, assume that $\bigwedge_j C_j$ is satisfiable and that $\phi : \{x_i\} \mapsto \{\top,\bot\}$ is a truth assignment that satisfies $\bigwedge_j C_j$, and construct $\lambda'$ as follows.
    \begin{align*}
        \forall i \quad & \lambda'(v_sv_{i,1}) = 1  \\
        \forall i \, \forall \ell \in [1,k-4] \quad & \lambda'(v_{i,\ell}v_{i,\ell+1}) = \ell+1  \\
        \forall i \quad & \lambda'(v_{i,k-4}v_{i,k-2}) = 2\delta+k-2 \\
        \forall i \quad & \lambda'(v_{i,k-1}v_{i,k}) = 2\delta+k \\
        \forall i,j \quad & \lambda'(v_{i,k-1}v^c_j) = \delta+k \\ 
        \forall i,j \quad & \lambda'(v_{i,k}v^c_j) = 2\delta+k+1 \\
        \forall i\, \text{ where }\phi(x_i) = \top \quad & \lambda'(v_{i,k-2}v_{i,k-1}) = k-1 \\
        \forall i\, \text{ where }\phi(x_i) = \bot \quad & \lambda'(v_{i,k-2}v_{i,k-1}) = 2\delta+k-1 \\
    \end{align*}
    This satisfies $\lambda'$ is a $(\delta,\zeta)$-perturbation of $\lambda$, and note that in particular, if $\phi(x_i) = \top$ then we have a $\top$-path for $x_i$ in $(G,\lambda')$, and if $\phi(x_i) = \bot$ then we have a $\bot$-path for $x_i$ in $(G,\lambda')$.
    
    We now consider shortest paths to various groups of vertices.
    First, for any variable $x_i$, if $\phi(x_i) = \top$ then $v_sv_{i,1}v_{i,2}\ldots v_{i,k-3}v_{i,k-2}v_{i,k-1}v_{i,k}$ is a temporal path of length $k$, and if $\phi(x_i) = \bot$ then $v_sv_{i,1}v_{i,2}\ldots v_{i,k-4}v_{i,k-2}v_{i,k-1}v_{i,k}$ is a temporal path of length $k-1$ and $v_sv_{i,1}v_{i,2}\ldots v_{i,k-4}v_{i,k-3}$ is a temporal path of length $k-3$.
    It thus remains to see whether vertices of the form $v^c_j$ are reachable in paths of length $k$.
    Take an arbitrary such vertex $v^c_j$; as $\phi$ satisfies $\bigwedge_j C_j$, at least one literal in $C_j$ must be true.
    If this literal is of the form $x_i$ (i.e., $x_i$ appears non-negated in $C_j$), then $\phi(x_i) = \top$ and so we have a $\top$-path for $x_i$ in $(G,\lambda')$ and thus the vertex $v^c_j$ is reachable by a path of length $k$.
    Next, if this literal is of the form $\overline {x_i}$ (i.e., $x_i$ appears negated in $C_j$), then $\phi(x_i) = \bot$ and so we have a $\bot$-path for $x_i$ in $(G,\lambda')$ and thus the vertex $v^c_j$ is reachable by a path of length $k$.
    Thus as $\lambda'$ is a $(\delta,m)$-perturbation, $((G, \lambda), v_s, k, \delta, m)$ is a yes-instance for TSEP.

    Next, assume that $((G,\lambda), v_s, k, 1, m)$ is a yes-instance for TSEP.
    Then there must be a $(\delta,m)$-perturbation $\lambda'$ of $\lambda$, such that for any clause vertex $v^c_j$, there is a path of length $k$ from $v_s$ to $v^c_j$.
    We now define $\phi: \{x_i\} \mapsto \{\top, \bot\}$. 
    For each edge of the form $v_{i,k-2}v_{i,k-1}$, if $\lambda'(v_{i,k-2}v_{i,k-1}) \leq \delta+k-1$ then let $\phi(x_i) = \top$, else let $\phi(x_i) = \bot$.
    Now consider an arbitrary clause $C_j$, and its associated clause vertex $v^c_j$.
    There must exist a temporal path of length $k$ in $(G,\lambda')$ from $v_s$ to $v^c_j$, and by the construction of $(G,\lambda)$ and the fact that $\lambda'$ is a $\delta$-perturbation, it must be that this path is either either a $\top$-path for some variable $x_i$, or a $\bot$-path for some variable $x_i$.
    If this path is a $\top$-path for some variable $x_i$ then $\lambda'(v_{i,k-2}v_{i,k-1} \leq \delta+k-1$ and so $\phi(x_i) = \top$, and by construction, the literal $x_i$ appears non-negated in $C_j$, so $C_j$ is satisfied.
    Alternatively, if this path is a $\bot$-path for some variable $x_i$ then $\lambda'(v_{i,k-2}v_{i,k-1} \geq 2\delta +k - 1 > \delta+k-1$ and so $\phi(x_i) = \bot$, and by construction, the literal $x_i$ appears negated in $C_j$, so $C_j$ is satisfied.
    Since we took an arbitrary clause vertex $v^c_j$, it follows that $\phi$ satisfies $\bigwedge_j C_j$.
\end{proof}

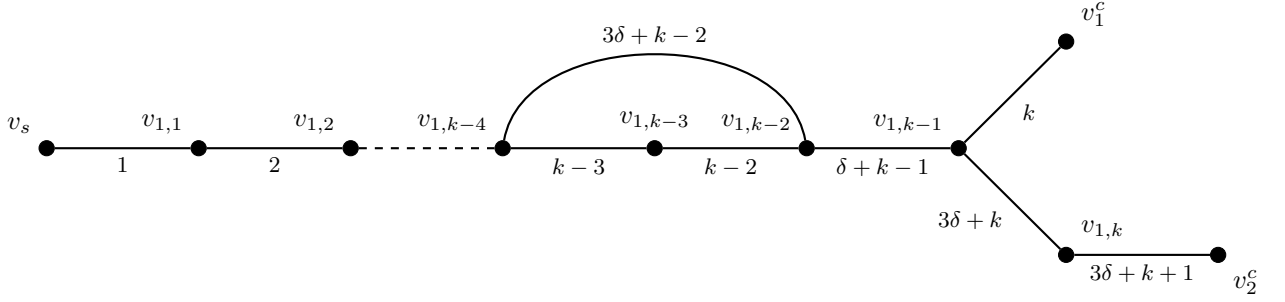
\begin{figure}
    \centering
    \begin{tikzpicture}[%
    node distance=1.8cm,
    vert/.style={draw, circle, minimum size=2mm, fill, inner sep=0pt},
    edge/.style={draw, thick},
    every edge quotes/.style = {auto, font=\footnotesize}
    ]
        \node[vert,label={above left:$v_s$}] (vs) {};
        \node[vert,label={above left:$v_{1,1}$}, right of=vs] (v11) {};
        \node[vert,label={above left:$v_{1,2}$}, right of=v11] (v12) {};
        \node[vert,label={above left:$v_{1,k-4}$}, right of=v12] (vk4) {};
        \node[vert,label={above:$v_{1,k-3}$}, right of=vk4] (vk3) {};
        \node[vert,label={above left:$v_{1,k-2}$}, right of=vk3] (vk2) {};
        \node[vert,label={above left:$v_{1,k-1}$}, right of=vk2] (vk1) {};
        \node[vert,label={above right:$v_{1,k}$}, below right of=vk1] (vk) {};
        \node[vert,label={above right:$v^c_1$}, above right of=vk1] (vc1) {};
        \node[vert,label={below right:$v^c_2$}, right of=vk] (vc2) {};
        
        \draw[edge] (vs) edge["{$1$}"'] (v11);
        \draw[edge] (v11) edge["{$2$}"'] (v12);
        \draw[edge] (v12) edge[dashed] (vk4);
        \draw[edge] (vk4) edge["{$k-3$}"'] (vk3);
        \draw[edge] (vk4) edge["{$3\delta+k-2$}", bend left=80] (vk2);
        \draw[edge] (vk3) edge["{$k-2$}"'] (vk2);
        \draw[edge] (vk2) edge["{$\delta+k-1$}"'] (vk1);
        \draw[edge] (vk1) edge["{$3\delta+k$}"'] (vk);
        \draw[edge] (vk1) edge["{$k$}"'] (vc1);
        \draw[edge] (vk) edge["{$3\delta+k+1$}"'] (vc2);
        
    \end{tikzpicture}
    
    \caption{Diagram of a partial construction for Theorem~\ref{theorem:tsep-many-perturbed-np-hard} showing the relevant vertices for variable $x_1$ and clauses $C_1$ and $C_2$.
    In this construction, clause $C_1$ contains the literal $x_1$ and clause $C_2$ contains the literal $\overline {x_1}$.
    This diagram is labelled with time labels from $\lambda$, and the dashed line indicates a path on the vertices $v_{i,2} \ldots v_{i,k-4}$, where $\lambda(v_{i,\ell}v_{i,\ell+1}) = \ell+2$.
    Note that, as $\delta=1$, if $\lambda'$ is a $1$-perturbation of $\lambda$ and contains a path of either length or duration $k$ from $v_s$ to $v^c_1$ through $v_{i,k-1}$ it must be that $\lambda'(v_{i,k-2}v_{i,k-1})=k-1$, while if $(G,\lambda')$ is to contain a path of length or duration $k$ from $v_s$ to $v^c_2$ through vertex $v_{i,k-1}$ it must that $\lambda'(v_{i,k-2}v_{i,k-1})=k$.
    This dichotomy translates to setting the variable $x_i$ to either true (if $\lambda'(v_{i,k-2}v_{i,k-1})=k-1$) or false (if $\lambda'(v_{i,k-2}v_{i,k-1})=k$).}
    \label{fig:tsep-many-perturbed-np-hard}
\end{figure}

We give an analogous result for TFaEP, which again involves a reduction from SAT.
\medskip
\begin{theorem}\label{theorem:tfaep-np-hard-any-k-delta}
    For any integers $k \geq 2$ and $\delta\geq 1$, TFaEP is \NP-hard even if $\zeta=m$.
\end{theorem}
\begin{proof}
    We show how, given integers $k\geq 2$ and $\delta$, an instance $((G,\lambda), k, \delta, m)$ (where $m=|E(G)|$) of TFaEP can be constructed from an instance $\bigwedge C_j$ of SAT such that $((G,\lambda), k, \delta, m))$ is a yes-instance of TFaEP if and only if $\bigwedge C_j$ is a yes-instance of SAT.
    First, we take our instance $\bigwedge C_j$ of SAT and construct our temporal graph $(G,\lambda)$, beginning by creating the vertex $v_s$.
    Then for each variable $x_i$, define $v_{i,0} = v_s$, and create the vertices $v_{i,\ell}$ for $\ell\in [1,k-1]$.
    Lastly, for each clause $C_j$ in the SAT instance create the vertex $v^c_j$.
    We next create the temporal edges.
    First, for each variable $x_i$ add the edge $(v_sv_{i,1}, 1)$.
    Then for each variable $x_i$, and for each $\ell\in [1,k-2]$, add the edge $(v_{i,\ell}v_{i,\ell+1}, \ell+1)$.
    Lastly, for each variable $x_i$, and each clause $C_j$, if $x_i$ appears as a literal (i.e., non-negated) in $C_j$, add the edge $(v_{i,k-1}v^c_j,k+2\delta )$,
    and if $\overline{x_i}$ appears as a literal in $C_j$, add the edge $(v_{i,k-1}v^c_j,k-1)$.
    A diagram of this gadget is shown in \Cref{fig:tfaep-np-hard-any-k-delta}.

    We now note some properties of this construction for an arbitrary $\delta$-perturbation $\lambda'$ of $\lambda$.
    First, consider some variable $x_i$, and some clause $C_j$.
    There can be no temporal path from $v_s$ to $v^c_j$ of length strictly less than $k$,
    and so for any temporal path from $v_s$ to $v^c_j$ in $(G,\lambda')$ of duration at most $k$, if edge $e'$ immediately follows edge $e$, then $\lambda'(e') = \lambda'(e)+1$.
    Now consider some variable $x_i$, and some clause $C_j$ such that $x_i$ only appears non-negated in $C_j$.
    To have a valid temporal path of duration $k$ in $(G,\lambda')$ from $v_s$ to $v^c_j$ through $v_{i,k-1}$, we must have $\lambda'(v_{i,k-2}v_{i,k-1}) = k-1 + \delta$ and $\lambda'(v_{i,k-1}v^c_j) = k + \delta$.
    On the other hand, if $x_i$ only appears negated in $C_j$, then for a valid temporal path of duration $k$ in $(G,\lambda')$ from $v_s$ to $v^c_j$ through $v_{i,k-1}$ to exist, we must have $\lambda'(v_{i,k-1}v^c_j) \leq k - 1 + \delta$, and so we must have $\lambda'(v_{i,k-2}v_{i,k-1}) \leq k - 2 + \delta$.
    We later use this dichotomy to encode whether the variable $x_i$ is true or false.

    We now show that $((G,\lambda),v_s, k, \delta, m)$ is a yes-instance to TFaEP if and only if $\bigwedge_j C_j$ is a yes-instance to SAT.
    First, assume that $\bigwedge_j C_j$ is satisfiable, and let $\phi$ be a function mapping variables to $\{\top,\bot\}$ that satisfies $\bigwedge_j C_j$.
    We will use $\phi$ to construct a a $\delta$-perturbation $\lambda'$ of $\lambda$.
    For each variable $x_i$, if $\phi(x_i) = \top$ then let $\lambda'(v_sv_{i,1}, 1+\delta)$, for each $\ell \in [1,k-2]$ let $\lambda'(v_{i,\ell}v_{i,\ell+1}) = \ell+1 + \delta$, and for each clause $C_j$ that contains $x_i$ as a literal (i.e.~non-negated) let $\lambda'(v_{i,k-1}v^c_j) = k + \delta$.
    For each variable $x_i$, if $\phi(x_i) = \bot$ then let $\lambda(v_sv_{i,1}, 1)$, for each $\ell \in [1,k-2]$ let $\lambda'(v_{i,\ell}v_{i,\ell+1}) = \ell+1$, and for each clause $C_j$ that contains $x_i$ as a literal (i.e.~non-negated) let $\lambda'(v_{i,k-1}v^c_j) = k$.
    We have not yet defined $\lambda'$ for some edges $e$ yet; for such edges, let $\lambda'(e) = \lambda(e)$.
    By this construction, $\lambda'$ is a $\delta$-perturbation of $\lambda$, so we now consider $\teccfa((G,\lambda'),v_s)$.
    First, for any variable $x_i$, the path $v_sv_{i,1}\ldots v_{i,k-2}v_{i,k-1}$ is a path of duration $k-1$ in $(G,\lambda')$, so we need only consider vertices of the form $v^c_j$.
    Consider an arbitrary such vertex $v^c_j$, and let $C_j$ be its associated clause. As $\phi$ satisfies $C_j$, there must be some variable $x_i$ such that either $\phi(x_i) = \top$ and $x_i$ appears as a literal in $C_j$, or $\phi(x_i) = \bot$ and $\overline{x_i}$ appears as a literal in $C_j$.
    In either case, the path $v_sv_{i,1}\ldots v_{i,k-2}v_{i,k-1}v^c_j$ is a temporal path of duration $k$ in $(G,\lambda')$, and so $\teccfa((G,\lambda'),v_s) = k$ and so $((G,\lambda), v_s, k, \delta, m)$ is a yes-instance to TFaEP.

    Next, assume that $((G,\lambda), v_s, k, \delta, m)$ is a yes-instance to TFaEP, and let $\lambda'$ be the associating perturbation that achieves $\teccfa((G,\lambda'),v_s) \leq k$.
    Construct a $\phi$ mapping variables to $\{\top,\bot\}$ as follows.
    If $\lambda'(v_{i,k-2}v_{i,k-1}) = k-1 + \delta$ let $\phi(x_i) = \top$, else let $\phi(x_i) = \top$.
    Now take an arbitrary clause $C_j$, and associated vertex $v^c_j$.
    There is a temporal path of duration $k$ in $(G,\lambda')$ from $v_s$ to $v^c_j$, as we have a yes-instance to TFaEP.
    By construction, such a path must use an edge of the form $v_{i,k-2}v_{i,k-1}$ for some $i$ corresponding to a variable $i$.
    If $\lambda'(v_{i,k-2}v_{i,k-1}) = k-1 + \delta$ then $\phi(x_i) = \top$ and the edge $v_{i,k-1}v^c_j$ must be used at time $k+\delta$, and so by construction (and as $\lambda'$ is a $\delta$-perturbation) it must be that $x_i$ appears non-negated in $C_j$, and so $C_j$ is satisfied.
    Alternatively, $\lambda'(v_{i,k-2}v_{i,k-1}) < k - 1 + \delta$ and $\phi(x_i) = \bot$, meaning we must have $\lambda'(v_{i,k-1}v^c_j) < k + \delta$, and so by construction (and as $\lambda'$ is a $\delta$-perturbation)  it must be that $x_i$ appears negated in $C_j$, and so $C_j$ is satisfied.
    As we chose $C_j$ arbitrarily, it follows that $\bigwedge_j C_j$ is satisfied by $\phi$, and so is a yes-instance to SAT.
\end{proof}

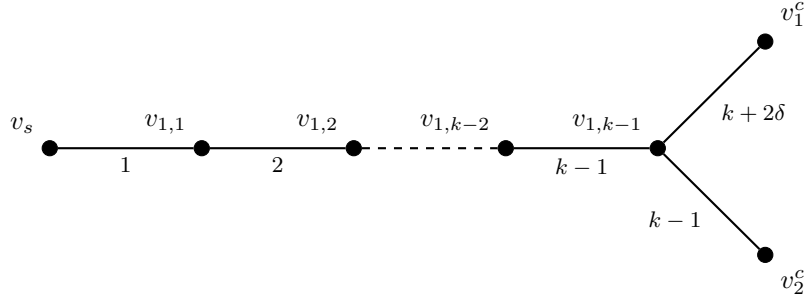
\begin{figure}
    \centering
    \begin{tikzpicture}[%
    node distance=2cm,
    vert/.style={draw, circle, minimum size=2mm, fill, inner sep=0pt},
    edge/.style={draw, thick},
    every edge quotes/.style = {auto, font=\footnotesize}
    ]
        \node[vert,label={above left:$v_s$}] (vs) {};
        \node[vert,label={above left:$v_{1,1}$}, right of=vs] (v11) {};
        \node[vert,label={above left:$v_{1,2}$}, right of=v11] (v12) {};
        \node[vert,label={above left:$v_{1,k-2}$}, right of=v12] (vk2) {};
        \node[vert,label={above left:$v_{1,k-1}$}, right of=vk2] (vk1) {};
        \node[vert,label={above right:$v^c_1$}, above right of=vk1] (vc1) {};
        \node[vert,label={below right:$v^c_2$}, below right of=vk1] (vc2) {};
        
        \draw[edge] (vs) edge["{$1$}"'] (v11);
        \draw[edge] (v11) edge["{$2$}"'] (v12);
        \draw[edge] (v12) edge[dashed] (vk4);
        \draw[edge] (vk2) edge["{$k-1$}"'] (vk1);
        \draw[edge] (vk1) edge["{$k + 2\delta$}"'] (vc1);
        \draw[edge] (vk1) edge["{$k-1$}"'] (vc2);
        
    \end{tikzpicture}
    \caption{Diagram of a partial construction for Theorem~\ref{theorem:tfaep-np-hard-any-k-delta} showing the relevant vertices for variable $x_1$ and clauses $C_1$ and $C_2$.
    In this construction, clause $C_1$ contains the literal $x_1$ and clause $C_2$ contains the literal $\overline {x_1}$.
    This diagram is labelled with time labels from $\lambda$, and the dashed line indicates a path on the vertices $v_{1,2} \ldots v_{1,k-4}$, where $\lambda(v_{1,\ell}v_{i,\ell+1}) = \ell+1$.
    Note that, as $\delta=1$, if $\lambda'$ is a $\delta$-perturbation of $\lambda$ and contains a path of duration at most $k$ from $v_s$ to $v^c_1$ through $v_{1,k-1}$ it must be that $\lambda'(v_{1,k-2}v_{1,k-1})=k+\delta-1$, while if $(G,\lambda')$ is to contain a path of duration at most $k$ from $v_s$ to $v^c_2$ through vertex $v_{1,k-1}$ it must be that $\lambda'(v_{1,k-2}v_{1,k-1})< k+\delta-1$.
    This dichotomy translates to setting the variable $x_1$ to either true (if $\lambda'(v_{1,k-2}v_{1,k-1})=k-1 + \delta$) or false (if $\lambda'(v_{1,k-2}v_{1,k-1}) \leq k-2+\delta$).}
    \label{fig:tfaep-np-hard-any-k-delta}
\end{figure}

We conclude this section with one positive result for TSEP and TFaEP.  Recall that TRLP becomes trivial when both $\zeta$ and $\delta$ are sufficiently large (every instance is a yes-instance).  While neither TSEP nor TFaEP becomes trivial under the same conditions, we now show that both problems can be solved in polynomial time when both $\zeta$ and $\delta$ are large.

Our algorithms make use of ubiquitous foremost temporal path trees.  Recall that 
a ubiquitous foremost temporal path is a foremost temporal path $P$ such that all prefixes of $P$ are themselves foremost temporal paths, and that a ubiquitous foremost temporal path tree from a source vertex $v_s$ is a temporal graph $(G_T,\lambda_T)$ such that $G_T$ is a tree, $\tau(G_T,\lambda_T) = 1$ and for any $v_t\in V(G)$, a foremost temporal path in $(G_T,\lambda_T)$ from $v_s$ to $v_t$ arrives at $v_t$ no later than any temporal path.  Our algorithms further rely on the following observation: for any edge $e$, the possible set of times it can be perturbed to includes the range $\{1,\ldots,r_e\}$ for some suitably large integer $r_e$, which means that along any sequence of (non-temporal) edges $(e_1,e_2,\ldots)$ that can form a temporal path under some suitable perturbation, we can assume that $\lambda(e_i) = i$.  We write $d_G(v_1,v_2)$ for the number of edges in the shortest path between $v_1$ and $v_2$ in $G$.

\medskip
\begin{theorem}\label{theorem:tsep-tfaep-if-large-delta-and-zeta}
    Each of TSEP and TFaEP can be solved in time $O(m(\log n + \log \delta \tau(G,\lambda)))$ on instances $((G,\lambda), v_s, k, \delta, \zeta)$ if $\delta \geq T(G,\lambda)$ and $\zeta \geq n - 1$.
\end{theorem}
\begin{proof}
    Define $\lambda'$ by $\lambda'(e) = \{ \max(1, \lambda(e) + p) \mid p \in [-\delta, \delta]\}$ for $e \in E(G)$, noting
    that as $\delta \geq T(G,\lambda)$, for each $e\in E(G)$ there is some positive integer $r_e$ such that $\lambda'(e) = \{1,2,\ldots, r_e\}$.
    Then use~\Cref{theorem:ubiquitous-foremost-tree} to calculate a ubiquitous foremost temporal path tree $(G_T, \lambda_T)$ on $v_s$ in the temporal graph $(G,\lambda')$
    in $O(m(\log n + \log\tau(G,\lambda'))) = O(m(\log n + \log \delta\tau(G,\lambda)))$ time.
    By setting $\lambda''(e) = \lambda_T(e)$ for $e\in E(G_T)$ and $\lambda''(e) = \lambda(e)$ for $e\not\in E(G_T)$, we know that $\lambda''$ is a $(\delta,\zeta)$-perturbation of $\lambda$ with the property that if a temporal path exists in $(G_T,\lambda_T)$, then the same path must also exist in $(G,\lambda'')$.
    Note that while $V(G_T)\subseteq V(G)$, it is possible that $V(G) \neq V(G_T)$: in this case the instance of either TSEP or TFaEP is a no-instance, so we assume from now on that $V(G_T) = V(G)$.
    We claim that for each goal vertex $v_g\in V(G)$, the unique temporal path from $v_s$ to $v_g$ in $(G_T,\lambda_T)$ arrives at time $d_{G_T}(v_s,v_g)$, where $d_{G_T}(v_s,v_g)$ is the distance from $v_s$ to $v_g$ in the graph $G_T$.

    We show this by induction on $d_{G_T}(v_s,v_g)$.
    If $d_{G_T}(v_s,v_g) = 1$, then $v_sv_g\in E(G_T)$, and as $1\in \lambda'(v_sv_g)$, there is a foremost path from $v_s$ to $v_g$ that arrives at time 1.
    By definition of a ubiquitous foremost temporal path tree, the unique temporal path from $v_s$ to $v_g$ must therefore arrive at time 1, completing the base case for the induction.

    We next proceed by induction.
    Let $d_{G_T}(v_s,v_g) = d > 1$, and let $v_1$ be a vertex such that $v_1v_g\in E(G_T)$ and $d_{G_T}(v_s,v_1) = d-1$.
    Such a vertex $v_1$ must exist, else there can be no path of length $d$ from $v_s$ to $v_g$ in $G_T$.
    By induction, the unique temporal path from $v_s$ to $v_1$ in $(G_T,\lambda_T)$ arrives at time $d-1$.
    Additionally, there is a temporal path from $v_s$ to $v_g$ in $(G_T,\lambda_T)$ (else we would have $v_g\not\in V(G_T)$.
    Let $d'$ be the arrival time of the unique foremost temporal path in $(G_T,\lambda_T)$ from $v_s$ to $v_g$.
    As $\lambda'(v_1v_g) = \{1,2,\ldots, r_e\}$ for some positive integer $r_e$, if $d' > d$ then $d'\in \{1,2,\ldots,r_e\}$ and so we must have $d\in \{1,2,\ldots,r_e\}$ as both are positive integers. But then the edge $v_1v_g$ could be used at time $d$ instead, breaking our assumption.
    Therefore the unique temporal path from $v_s$ to $v_g$ in $(G_T,\lambda_T)$ arrives at time $d$.

    Noting that, as $(G_T,\lambda_T)$ is a ubiquitous foremost temporal path tree, no temporal path in any $(\delta,\zeta)$-perturbation of $\lambda$ can arrive earlier than $d_{G_T}(v_s,v_g)$, 
    we can also deduce that no temporal path in any $(\delta,\zeta)$-perturbation could be shorter than $d_{G_T}(v_s,v_g)$ (else it could arrive earlier than $d_{G_T}(v_s,v_g)$) and no temporal path in any $(\delta,\zeta)$-perturbation can have a lower duration (by the same argument).
    It therefore remains to check that either $\teccs{((G_T,\lambda_T), v_s)} \leq k$ (if we are considering TSEP) or $\teccfa{((G_T,\lambda_T), v_s)} \leq k$ (if we are considering TFaEP).
    As $(G_T,\lambda_T)$ is a ubiquitous foremost temporal path tree, each of those can be checked in $O(n)$ time.
\end{proof}
\section{Conclusion and Open Questions}\label{sec:conclusion}
Motivated by the uncertainty and errors intrinsic in recording real-world temporal graphs, we have investigated temporal graph reachability problems under perturbation, in particular asking whether there exists a perturbation of time-edges of a temporal graph that results in a minimum required reachability.  We find that when arbitrarily many perturbations are allowed, this can be resolved efficiently, but that the problem is, in general, \NP-hard and $W[2]$-hard with respect to $\zeta$. Furthermore, under the assumption of ETH, we show that there is no $f(\zeta)n^{o(\zeta)}$-time algorithm for our problem. We have tractability results when the underlying graph is a tree or when it has bounded treewidth and the lifetime and size of permitted perturbations are also bounded by constants.  For contrast, we show that allowing arbitrarily many perturbations is not sufficient to give tractability of related eccentricity problems, where the question is not just whether there is some source vertex that reaches some specified number of vertices, but whether it reaches vertices within some travel duration (fastest) or in few steps (shortest).  

An obvious avenue for further research is on other temporal graph problems with the same style of perturbations.  However, inspired by the motivation to make temporal algorithms more suitable for real-world temporal contact data, we suggest exploring temporal graph modification problems that are robust to  perturbations: for example, removing (by, e.g. vaccination) vertices from a temporal network to limit reachability in a way that has guarantees under possible bounds of perturbation.  


\backmatter

\section*{Statements and Declarations}
The authors have no competing interests and no data was used in this work.
For the purpose of open access, the author(s) has applied a Creative Commons Attribution (CC BY) licence to any Author Accepted Manuscript version arising from this submission.

\subsection*{Author Contributions}
All authors (Jessica Enright, Laura Larios-Jones, William Pettersson, Kitty Meeks) contributed to the generation of ideas, the writing of the manuscript and the reviewing of the manuscript.

\subsection*{Funding}
Jessica Enright, Kitty Meeks, and William Pettersson are supported by EPSRC grant EP/T004878/1. William Pettersson is additionally supported by EPSRC grant EP/X013618/1. 


\newpage
\clearpage
\bibliography{reachabilityRefs}
\end{document}